%% file: ms.tex
\documentclass[twoside,leqno,twocolumn]{article}

\usepackage{ltexpprt}

\usepackage[utf8]{inputenc}
\usepackage{mathptmx}

\usepackage{paralist,tabularx}
\usepackage{amsmath,amsfonts,amssymb}
\usepackage{latexsym}
\usepackage{enumerate}
\usepackage{mathtools}
\usepackage{wrapfig}

\newtheorem{hypothesis}{Hypothesis}[section]

\newenvironment{definition}[1][Definition]{\begin{trivlist}
		\item[\hskip \labelsep {\bfseries #1}]}{\end{trivlist}}
	
\newenvironment{reminder}[1]{\smallskip
	\noindent {\scshape Reminder of #1 }\em}{
}

\newcommand{\ttildeo}{\tilde{o}}
\newcommand{\tO}{\tilde{O}}
\newcommand{\tTheta}{\tilde{\Theta}}
\newcommand{\tOmg}{\tilde{\Omega}}

\newcommand{\poly}{\text{poly}}

\newcommand{\eps}{\epsilon}

\DeclarePairedDelimiter{\ceil}{\lceil}{\rceil}

\begin{document}
\title{Tight Hardness for Shortest Cycles and Paths in Sparse Graphs} 
\author{Andrea Lincoln\thanks{andreali@mit.edu. Supported by the EECS Merrill Lynch Fellowship.} \and Virginia Vassilevska Williams\thanks{virgi@csail.mit.edu. Supported by an NSF CAREER Award, NSF Grants CCF-1417238, CCF-1528078 and CCF-1514339, and BSF Grant BSF:2012338.} \and Ryan Williams\thanks{rrw@mit.edu. Supported by an NSF CAREER Award.}}
\date{}
\maketitle


\begin{abstract}
\input{abstract}
\end{abstract}

\thispagestyle{empty}
\newpage

\setlength{\parskip}{0.5ex plus 0.25ex minus 0.25ex}
\setlength{\parindent}{10pt}

\section{Introduction}
\input{intro-thmsVersion1}
\section{Preliminaries}
\label{sec:Prelim}
\input{prelims}


\section{Reduction from Hyperclique to Hypercycle}
\label{sec:redFromHyperCliquetoHypercycle}
\input{hypercliquetohypercycle}

\section{Reduction from Hypercycle to Cycle in Directed Graphs}
\label{sec:hyperToCycle}

\input{hypercycleDirectedcycle}

\section{Probably Optimal Weighted k-Cycle Algorithms}
\label{sec:cycleUpperBounds}
\input{shortestCycleUB}

\section{Hardness Results for Shortest Cycle}
\label{sec:shortestCycleHard}
\input{shortestkvsShort}

\section{Discussion of the Hyperclique Hypothesis}
\label{sec:hypothesisDiscussion}
\input{hyperdiscuss}

\section{No Generalized Matrix Multiplication for k$>$2}
\label{sec:genMatrixMult}
\input{nogeneralizedmmult}

\section{Max-k-SAT to Tight Hypercycle}
\label{sec:MaxKSatHypercycle}
\input{hypergraphsmaxSAT}

\paragraph{Acknowledgments.} V.V.W. would like to acknowledge J.M. Landsberg and Mateusz Michalek for valuable discussions about tensors rank lower bounds. We would like to thank the anonymous reviewers whose suggestions we implemented. We thank Pawel Gawrychowski for pointing out a typo in a previous version of the paper.

\let\realbibitem=\bibitem
\def\bibitem{\par \vspace{-0.5ex}\realbibitem}

\bibliographystyle{alpha}
\bibliography{mybib} 

\newpage

\appendix
\section{Prior Work On Improved Running Times}
\label{sec:relatedwork}
\input{relatedwork}

\section{Beating O(mn) for Many Graph Problems is Probably Hard}\label{sec:newreducs}
\input{radius}
\input{wienerIndex}
\input{APSP}

\section{General CSP to Hyperclique}
\label{sec:generalCSPtoHyperclique}
\input{GeneralCSP}

\section{Bounds for Graph Densities for NonInteger L}
\label{sec:nonIntgerDensities}
\input{nonIntDensity}

\end{document}

%% file: abstract.tex
Fine-grained reductions have established equivalences between many core  problems with $\tilde{O}(n^3)$-time algorithms on $n$-node weighted graphs, such as Shortest Cycle, All-Pairs Shortest Paths (APSP), Radius, Replacement Paths, Second Shortest Paths, and so on. These problems also have $\tilde{O}(mn)$-time algorithms on $m$-edge $n$-node weighted graphs, and such algorithms have wider applicability. Are these $mn$ bounds optimal when $m \ll n^2$?

Starting from the hypothesis that the minimum weight $(2\ell+1)$-Clique problem in edge weighted graphs requires $n^{2\ell+1-o(1)}$ time, we prove that for all sparsities of the form $m = \Theta(n^{1+1/\ell})$, there is no $O(n^2 + mn^{1-\eps})$ time algorithm for $\epsilon>0$ for \emph{any} of the below problems 
\begin{compactitem}
\item Minimum Weight $(2\ell+1)$-Cycle in a directed weighted graph,
\item Shortest Cycle in a directed weighted graph,
\item APSP in a directed or undirected weighted graph,
\item Radius (or Eccentricities) in a directed or undirected weighted graph,
\item Wiener index of a directed or undirected weighted graph,
\item Replacement Paths in  a directed weighted graph,
\item Second Shortest Path in a directed weighted graph,
\item Betweenness Centrality of a given node in a directed weighted graph.
\end{compactitem}
That is, we prove hardness for a variety of sparse graph problems from the hardness of a dense graph problem. Our results also lead to new conditional lower bounds from several related hypothesis for unweighted sparse graph problems including $k$-cycle, shortest cycle, Radius, Wiener index and APSP. 

%% file: intro-thmsVersion1.tex

The All-Pairs Shortest Paths (APSP) problem is among the most basic computational problems. A powerful primitive, APSP can be used to solve many other problems on graphs (e.g. graph parameters such as the girth or the radius), but also many non-graph problems such as finding a subarray of maximum sum~\cite{TamakiT98} or parsing stochastic context free grammars (e.g.~\cite{Akutsu1999}). 
 Over the years, many APSP algorithms have been developed. For edge weighted $n$ node, $m$ edge graphs,
the fastest known algorithms run in $n^3/2^{\Theta(\sqrt{\log n})}$ time~\cite{williams2014faster} for dense graphs, and in $O(mn+n^2\log\log n)$ time~\cite{pettie2002faster} for sparse graphs. 

These running times are also essentially the best known for many of the problems that APSP can solve: Shortest Cycle, Radius, Median, Eccentricities, Second Shortest Paths, Replacement Paths, and so on\footnote{For Shortest Cycle, an $O(mn)$ time algorithm was recently developed by Orlin and Sede{\~{n}}o{-}Noda~\cite{OrlinS17}. For a full discussion of the best known running times of these problems see Appendix \ref{sec:relatedwork}.}.
For dense graphs, this was explained by Vassilevska Williams and Williams~\cite{williams2010subcubic} and later Abboud et al.~\cite{radiusgrandoni} who showed that either all of $\{$APSP, Minimum Weight Triangle, Shortest Cycle, Radius, Median, Eccentricities, Second Shortest Paths, Replacement Paths, Betweenness Centrality$\}$ have {\em truly subcubic} algorithms (with runtime $O(n^{3-\eps})$ for constant $\eps>0$), or none of them do. Together with the popular hypothesis that APSP requires $n^{3-o(1)}$ time on a word-RAM (see e.g.~\cite{abboud2014popular,radiusgrandoni,ipecsurvey,treeedit}), these equivalences suggest that all these graph problems require $n^{3-o(1)}$ time to solve.

However, these equivalences no longer seem to hold for sparse graphs. The running times for these problems still match: $\tilde{O}(mn)$ is the best running time known for all of these problems. In recent work, Agarwal and Ramachandran~\cite{agarwal2016fine} show that some reductions from prior work can be modified to preserve sparsity. Their main result is that if Shortest Cycle in directed weighted graphs requires $m^a n^{2-a-o(1)}$ time for some constant $a$, then so do Radius, Eccentricities, Second Shortest Paths, Replacement Paths, Betweenness Centrality and APSP in directed weighted graphs\footnote{This is analogous to the dense graph regime of \cite{williams2010subcubic}, where the main reductions went from Minimum Weight 3-Cycle (i.e. triangle). The key point of \cite{agarwal2016fine} is that one can replace Minimum Weight 3-Cycle by Minimum Weight Cycle, and preserve the sparsity in the reduction.}.

Unfortunately, there is no known reduction that preserves sparsity from APSP (or any of the other problems) back to Shortest Cycle, and there are no known reductions to Shortest Cycle from any other problems used as a basis for hardness within Fine-Grained Complexity, such as the Strong Exponential Time Hypothesis~\cite{ipz1,ipz2}, $3$SUM~\cite{GO95} or Orthogonal Vectors~\cite{Will05,ipecsurvey}. Without a convincing reduction, one might wonder: 
\begin{center}{\em Can Shortest Cycle in weighted directed graphs be solved in, say, $\tO(m^{3/2})$  time?\\ Can APSP be solved in $\tO(m^{3/2}+n^2)$  time?}\end{center}
Such runtimes are consistent with the dense regime of $m=\tilde{\Theta}(n^2)$. Minimum Weight Triangle, which is the basis of many reductions in the dense case, can be solved in $O(m^{3/2})$ time (e.g. \cite{itai1978finding}). What prevents us from having such running times for all the problems that are equivalent in the dense regime to Minimum Weight Triangle?
  Why do our best algorithms for these other problems take $\tO(mn)$ time, and no faster? In fact, we know of no $\eps>0$ for which problems like Shortest Cycle can be solved in $\tO(m^{1+\eps}n^{1-2\eps})$ time. Such a running time is $\Theta(n^3)$ for $m=\Theta(n^2)$ and $o(mn)$ for $m\leq o(n^2)$. Notice that $m^{3/2}$ is the special case for $\eps=1/2$. Is there a good reason why no  $\tO(m^{1+\eps}n^{1-2\eps})$ time algorithms have been found?

\paragraph{Our results.}
We give compelling reasons for the difficulty of improving over $\tilde{O}(mn)$ for Shortest Cycle, APSP and other problems. 
We show for an infinite number of sparsities, any sparsity $m=n^{1+1/\ell}$ where $\ell \in \mathbb{N}$, obtaining an $O(n^2+mn^{1-\eps})$ time algorithm for Shortest Cycle (or any of the other fundamental problems) in weighted graphs for any constant $\eps>0$ would refute a popular hypothesis about the complexity of weighted $k$-Clique. 

\begin{hypothesis}[Min Weight $k$-Clique] 
\label{hyp:minkClique}	There is a constant $c>1$ such that,
on a Word-RAM with $O(\log n)$-bit words, finding a $k$-Clique of minimum total edge weight in an $n$-node graph with nonnegative integer edge weights in $[1,n^{ck}]$ requires $n^{k-o(1)}$ time.
\end{hypothesis}

The Min Weight $k$-Clique Hypothesis has been considered for instance in \cite{BackursT16} and \cite{abboud2014consequences} to show hardness for improving upon the Viterbi algorithm, and for Local Sequence Alignment.
The (unweighted) $k$-Clique problem is NP-Complete, but can be solved in $O(n^{\omega k/3})$ time when $k$ is fixed~\cite{NesetrilPoljak}\footnote{When $k$ is divisible by $3$; slightly slower otherwise.} where $\omega<2.373$~\cite{vstoc12,legall} is the matrix multiplication exponent.
The problem is W[1]-complete and under the Exponential Time Hypothesis~\cite{ipz1} it cannot be solved in $n^{o(k)}$ time. 
Finding a $k$-Clique of minimum total weight (a Min Weight $k$-Clique) in an edge-weighted graph can also be solved in $O(n^{\omega k/3})$ if the edge weights are small enough. However, when the edge weights are integers larger than $n^{ck}$ for large enough constant $c$, the fastest known algorithm for Min Weight $k$-Clique runs in essentially $O(n^k)$ time (ignoring $n^{o(1)}$ improvements). The special case $k=3$, Min Weight $3$-Clique is the aforementioned Minimum Weight Triangle problem which is equivalent to APSP under subcubic reductions and is believed to require $n^{3-o(1)}$ time.

Building on Vassilevska Williams and Williams~\cite{williams2010subcubic}, Agarwal and Ramachandran~\cite{agarwal2016fine} have shown many sparsity-preserving reductions from Shortest Cycle to various fundamental graph problems. They thus identify Shortest Cycle as a fundamental bottleneck to improving upon $mn$ for many problems. However, so far there is no compelling reason why Shortest Cycle itself should need $mn$ time.

\begin{theorem}[\cite{agarwal2016fine,williams2010subcubic}] 
	Suppose that there is a constant $\eps>0$ such that one of the following problems on $n$-node, $m$-edge weighted graphs can be solved in
 $O(mn^{1-\eps}+n^2)$ time: 
\begin{compactitem}
\item APSP in a directed weighted graph,
\item Radius (or Eccentricities) in a directed weighted graph,
\item Replacement Paths in a directed weighted graph,
\item Second Shortest Path in a directed weighted graph,
\item Betweenness Centrality of a given node in a directed weighted graph.
\end{compactitem}
Then, the Min Weight Cycle Problem is solvable in $O(mn^{1-\eps'}+n^2)$ for some $\eps'>0$ time~\cite{agarwal2016fine}.\label{cor:direcHard}
\end{theorem}

Our main technical contribution connects the complexity of small cliques in dense graphs to that of small cycles in sparse graphs:

\begin{theorem} Suppose that there is an integer $\ell\geq 1$ and a constant $\eps>0$ such that one of the following problems on $n$-node, $m=\Theta(n^{1+1/\ell})$-edge weighted graphs can be solved in
	$O(mn^{1-\eps}+n^2)$ time: 
	\begin{compactitem}
		\item Minimum Weight $(2\ell+1)$-Cycle in a directed weighted graph,
		\item Shortest Cycle in a directed weighted graph,
	\end{compactitem}	
	Then, the Min Weight $(2\ell+1)$-Clique Hypothesis is false.\label{thm:girthHard}
\end{theorem}

Combining our main Theorem~\ref{thm:girthHard} with the results from previous work in Theorem~\ref{cor:direcHard} gives us new conditional lower bounds for fundamental graph problems. We also create novel reductions from the $k$-Cycle problem, in Section~\ref{sec:newreducs}, and these give us novel hardness results for many new problems. 
The main new contributions are reductions to Radius in undirected graphs (the result in \cite{agarwal2016fine} is only for directed) and to the Wiener Index problem which asks for the sum of all distances in the graph.
Together all these pieces give us the following theorem. 

\begin{theorem} Suppose that there is an integer $\ell\geq 1$ and a constant $\eps>0$ such that one of the following problems on $n$-node, $m=\Theta(n^{1+1/\ell})$-edge weighted graphs can be solved in
 $O(mn^{1-\eps}+n^2)$ time: 
\begin{compactitem}
\item Minimum Weight $(2\ell+1)$-Cycle in a directed weighted graph,
\item Shortest Cycle in a directed weighted graph,
\item APSP in a directed or undirected weighted graph,
\item Radius (or Eccentricities) in a directed or undirected weighted graph,
\item Wiener index of a directed or undirected weighted graph,
\item Replacement Paths in a directed weighted graph,
\item Second Shortest Path in a directed weighted graph,
\item Betweenness Centrality of a given node in a directed weighted graph.
\end{compactitem}
Then, the Min Weight $(2\ell+1)$-Clique Hypothesis is false.\label{thm:weightedhard}
\end{theorem}

So, either min weighted cliques can be found faster, or $\tilde{O}(mn+n^2)$ is the {\em optimal} running time for these problems, up to $n^{o(1)}$ factors, for an infinite family of edge sparsities $m(n)$. 
See Figure~\ref{fig:lbimg2} in the Appendix for a pictorial representation of our conditional lower bounds.  

Another intriguing consequence of Theorem~\ref{thm:weightedhard} is that, assuming Min Weight Clique is hard, running times of the form $\tO(m^{1+\eps}n^{1-2\eps})$ for $\eps>0$ are impossible! If Shortest Cycle had such an algorithm for $\eps>0$, then for every integer $L>1$ and $\delta=(1-1/L)\eps>0$ we have that for $m=\Theta(n^{1+1/L})$, $m^{1+\eps}n^{1-2\eps}\leq mn^{1-\delta}$ and hence the Min Weight $(2L+1)$-Clique Hypothesis is false.

\input{weightedTable}

Our reduction from Minimum Weight $(2L+1)$-Clique to Minimum Weight $(2L+1)$-Cycle produces a directed graph on $n^{L}$ nodes and $m=O(n^{L+1})$ edges, and hence if directed Minimum Weight $(2L+1)$-Cycle can be solved in $O(m^{2-1/(L+1) -\eps})$ time for some $\eps>0$, then the Min Weight $(2L+1)$-Clique Hypothesis is false. We present an extension for weighted cycles of even length as well, obtaining:

\begin{corollary} 
If Minimum Weight $k$-Cycle  in directed $m$-edge graphs is solvable in $O(m^{2-2/(k+1)-\eps})$ time for some $\eps>0$ for $k$ odd, or in $O(m^{2-2/k-\eps})$ time for $k$ even, then the Minimum Weight $\ell$-Clique Hypothesis is false for $\ell = 2\lceil k/2\rceil-1$.
\end{corollary}

Directed $k$-cycles in unweighted graphs were studied by Alon, Yuster and Zwick~\cite{alon1997finding} who gave algorithms with a runtime of $O(m^{2-2/(k+1)})$ for $k$ odd, and $O(m^{2-2/k})$ for $k$ even. We show that their algorithm can be extended to find Minimum Weight $k$-Cycles with only a polylogarithmic overhead, proving that the above conditional lower bound is tight.

\begin{theorem}
The Minimum Weight $k$-Cycle in directed $m$-edge graphs can be solved in $\tilde{O}(m^{2-2/(k+1)})$ time for $k$ odd, and in $\tilde{O}(m^{2-2/k})$ time for $k$ even.
\end{theorem}

Our lower bound results compared to prior work are presented in
Table \ref{table:weighted}. 
The upper bounds for the considered problems are as follows: Min $k$-Clique is easily solvable in $O(n^k)$ time. The best algorithms for all other problems in the table take $\tO(mn)$ time~\cite{AlonYZ16,OrlinS17,pettie2002faster,dijkstra,gotthilf2009improved}. 


\begin{figure*}[h]
	\centering
	\includegraphics[scale=0.8]{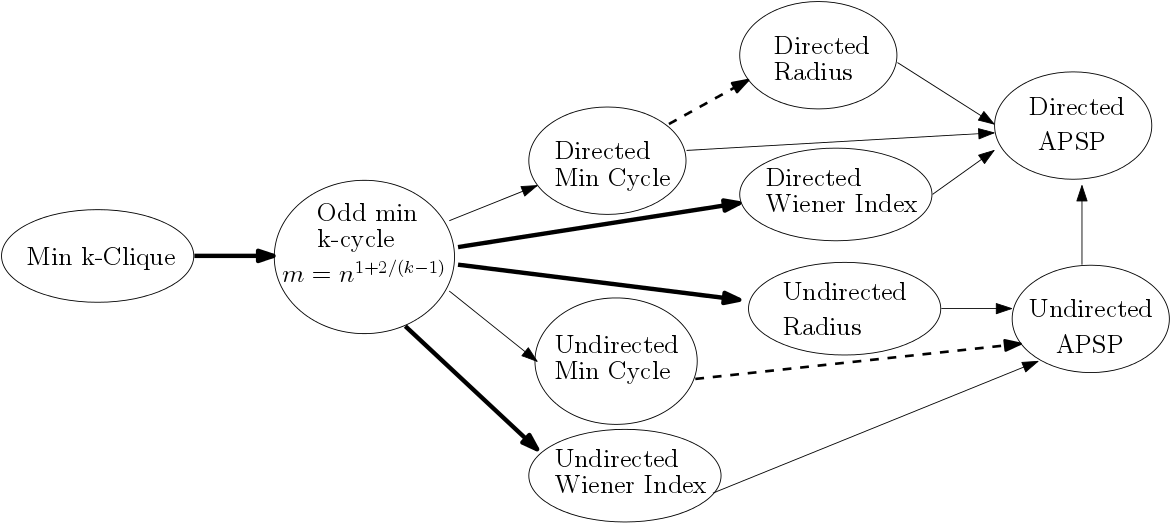}
	\caption{A depiction of a network of reductions related to sparse weighted graph problems and the dense Minimum $k$-clique problem. The bold edges represent a subset of the reductions in this paper. The green edges are reductions from Agarwal and Ramachandran~\cite{agarwal2016fine}.}
	\label{fig:reductionDiagram}
\end{figure*}

\input{sparseIntro}

%% file: weightedTable.tex
\begin{table*}[h]
$
\begin{array}{|c|c|c|c|c|c|} 

\hline
\text{Weighted Problem}  &\text{Lower Bnd} &  \text{LB from} & \text{LB source} \\
\hline
\text{Min k-clique}&n^{k-o(1)}&\text{$k$-clique conj.}&\text{By Def}\\
\text{Min k-cycle (k odd)}& mn^{1-o(1)}&\textbf{Min k-clique}&\textbf{Thm \ref{lem:directedCycleHyperCycle}}\\
\text{Shortest cycle}&mn^{1-o(1)}&\textbf{odd min k-cycle}&\textbf{Thm \ref{lem:mincyclekcycle}} \\
\text{Directed Radius}&mn^{1-o(1)}&\text{Shortest Cycle}&\text{\cite{agarwal2016fine}} \\
\text{Undirected Radius}&mn^{1-o(1)}&\textbf{odd min k-cycle}&\textbf{Thm \ref{thm:undRadiusisMN}} \\
\text{Directed APSP}&mn^{1-o(1)}&\text{Shortest Cycle}&\text{\cite{agarwal2016fine}} \\
\text{Undirected APSP}&mn^{1-o(1)}&\textbf{odd min k-cycle}&\textbf{Cor \ref{cor:APSPhard}} \\
\text{Undirected APSP}&mn^{1-o(1)}&\text{Und. Shortest Cycle}& \text{\cite{agarwal2016fine}}\\
\text{Und. Wiener Index}&mn^{1-o(1)}&\textbf{odd min k-cycle}&\textbf{Thm \ref{thm:wienerIndex}} \\
\text{Dir. } 2^{nd} \text{ Shortest Path}&mn^{1-o(1)}&\text{Shortest Cycle}&\text{\cite{agarwal2016fine}} \\
\text{Dir. Repl. Paths}&mn^{1-o(1)}&\text{Shortest Cycle}&\text{\cite{agarwal2016fine}} \\
\hline
\end{array}
$
\centering
\caption{Weighted graph lower bounds. Our results are in bold. Und stands for undirected and Dir stands for directed. Repl stands for replacement.}
\label{table:weighted}
\end{table*}

%% file: sparseIntro.tex
\paragraph{Sparse Unweighted Problems.}

We have proven tight conditional lower bounds for weighted graphs.
However, for sparse enough ($m=O(n^{\omega-1})$) {\em unweighted} graphs, the best algorithms for APSP and its relatives also run in $\tilde{O}(mn)$ time (see Section~\ref{sec:relatedwork} for the relevant prior work on APSP).
We hence turn our attention to the unweighted versions of these problems.

Our reduction from Min Weight $k$-Clique to Min Weight $k$-Cycle still works for unweighted graphs just by disregarding the weights. We can get super-linear lower bounds for sparse unweighted problems from three different plausible assumptions. 

As  mentioned before, the fastest algorithm for $k$-Clique (for $k$ divisible by $3$) runs in $O(n^{\omega k/3})$~\cite{NesetrilPoljak,fixedParamClique}. This algorithm has remained unchallenged for many decades and lead to the following hypothesis (see e.g. \cite{AbboudRNA}).

\begin{hypothesis}[The $k$-Clique Hypothesis]
Detecting a $k$-Clique in a graph with $n$ nodes requires $\Omega (n^{\omega k/3-o(1)})$ time on a Word RAM. 
\label{hyp:UnwKCliqe}
\end{hypothesis}

From this we get super-linear lower bound for the shortest cycle problem. We get an analogous result to the one we had before:


\begin{theorem}If the $k$-Clique Hypothesis is  true, Shortest Cycle in undirected or directed graphs requires $m^{2\omega/3-o(1)}$.
\label{thm:sparseShortCycle}
\end{theorem}

We get super-linear lower bounds for various graph problems as a corollary of Theorem~\ref{thm:sparseShortCycle}. 

\begin{corollary}
	If the $k$-Clique Hypothesis is  true, the following problems in unweighted graphs require $m^{2\omega/3-o(1)}$ time:
	\begin{compactitem}
		\item Betweenness Centrality in a directed graph,
		\item APSP in an undirected or directed graph,
		\item Radius in an undirected or directed graphs, n
		\item Wiener Index in an undirected or directed graph.
	\end{compactitem}
\end{corollary}

The reader may notice that the matrix multiplication exponent shows up repeatedly in the unweighted cases of these problems. This is no coincidence.  The best known combinatorial\footnote{Informally, combinatorial algorithms are algorithms that do not use fast matrix multiplication.} algorithms for the unweighted $k$-clique algorithm take $\tTheta(n^k)$ time. This has lead to a new hypothesis.

\begin{hypothesis}[Combinatorial $k$-Clique]
	Any \emph{combinatorial} algorithm to detect a $k$-Clique in a graph with $n$ nodes requires $n^{k-o(1)}$ time on a Word RAM~\cite{AbboudRNA}. 
	\label{hyp:UnwCombKCliqe}
\end{hypothesis}

Our reduction from $k$-clique to $k$-cycle is combinatorial. Thus, an $O(mn^{1-\eps})$ time (for $\eps>0$) combinatorial algorithm for the directed unweighted $k$-cycle problem for odd $k$ and $m= n^{(k+1)/(k-1)}$ would imply a combinatorial algorithm for the $k$-clique problem with running time $O(n^{k-\eps'})$ for $\eps'>0$. 
Any algorithm with a competitive running time must use fast matrix multiplication, or give an exciting new algorithm for $k$-clique.

Currently, the best bound on $\omega$ is $\omega<2.373$ \cite{legall,vstoc12}, and the $m^{2\omega/3-o(1)}$ lower bound for Shortest Cycle and related problems might conceivably be $m^{1.58}$ which is not far from the best known running time $O(m^{1.63})$ for $5$-cycle~\cite{yuster2004detecting}. Yuster and Zwick gave an algorithm based on matrix multiplication for directed $k$-Cycle, however they were unable to analyze its running time for $k>5$. 
They conjecture that if $\omega=2$, their algorithm runs faster than the best combinatorial algorithms for every $k$, however even the conjectured runtime goes to $m^2$, as $k$ grows. In contrast, for $\omega=2$, our lower bound based on the $k$-Clique Hypothesis is only $m^{4/3-o(1)}$\footnote{Of course, the Yuster and Zwick running time for $k$-cycle might not be optimal, and it might be that $O(m^{2-\delta})$ time is possible for $k$-cycle for some $\delta>0$ and all $k$.}. We thus search for higher conditional lower bounds based on different and at least as believable hypotheses.

To this end, we formalize a working hypothesis 
about the complexity of finding a hyperclique in a hypergraph. An $\ell$-hyperclique in a $k$-uniform hypergraph $G$ is composed of a set of $\ell$ nodes of $G$ such that all $k$-tuples of them form a hyperedge in $G$.

\begin{hypothesis}[$(\ell,k)$-Hyperclique Hypothesis]
	\label{hyp:lkhypercliqef} 
	Let $\ell>k>2$ be integers. On a Word-RAM with $O(\log n)$ bit words, finding an $\ell$-hyperclique in a $k$-uniform hypergraph on $n$ nodes requires $n^{\ell-o(1)}$ time.
\end{hypothesis}

Why should one believe the hyperclique hypothesis?
There are many reasons: (1) When $k>2$, there is no  $O(n^{\ell-\eps})$ time algorithm for any $\eps>0$ for $\ell$-hyperclique in $k$ uniform hypergraphs.
(2)
The natural extension of the techniques used to
solve $\ell$-clique in graphs will {\em NOT} solve $\ell$-hyperclique in $k$-uniform hypergraphs in $O(n^{\ell-\eps})$ time for any $\eps>0$ when $k>2$. We prove this in Section~\ref{sec:genMatrixMult}.
(3) There are known reductions from notoriously difficult problems such as Exact Weight $k$-Clique, Max $k$ SAT and even harder Constrained Satisfaction Problems (CSPs) to $(\ell,k)$-Hyperclique so that if the hypothesis is false, then all of these problems have exciting improved algorithms. For these and more, see the Discussion in Section~\ref{sec:discuss}.

Now, let us state our results for unweighted Shortest Cycle based on the $(\ell,k)$-Hypothesis. The same lower bounds apply to the other problems of consideration (APSP, Radius etc.). 

\begin{theorem}
Under the $(\ell,k)$-Hypothesis, the Shortest Cycle problem in directed unweighted graphs requires $m^{k/(k-1) - o(1)}$ time on a Word RAM with $O(\log n)$ bit words.
\end{theorem}

The theorem implies in particular that, Shortest Cycle in unweighted directed graphs requires (a) $m^{3/2-o(1)}$ time, unless Max $3$-SAT (and other CSPs) have faster than $2^n$ algorithms, (b) $m^{4/3-o(1)}$ time, unless Exact Weight $k$-Clique has a significantly faster than $n^k$ algorithm.
The latter is the same lower bound as from $k$-Clique when $\omega=2$ but it is from a different and potentially more believable hypothesis. Finally, Shortest Cycle and its relatives are not in linear time, unless the $(\ell,k)$-Hypothesis is false for every constant $k$.

\begin{table*}[h]
	$
	\begin{array}{|c|c|c|c|c|c|} 
	
	\hline
	\text{Problem}  &\text{Lower Bound} &  \text{LB from} & \text{LB src} \\
	\hline
	\text{k-cycle odd}& m^{(k-\ceil{(k+1)/3}+1)/k -o(1)}&\textbf{Max 3-SAT}&\textbf{Thm \ref{thm:dirCycleMaxSat}}\\
	\text{k-cycle odd} & m^{(2\omega k)/(3(k+1))-o(1)}&\textbf{k-clique}&\textbf{Lem  \ref{lem:directedCycleHyperCycle}}\\
	
	\text{ Shrt. cycle}&  m^{2\omega/3-o(1)}&\textbf{k-clique}&\textbf{Lem \ref{lem:minunweightedcyclekcycle}}\\
	\text{ Shrt. cycle}& m^{3/2 -o(1)}&\textbf{Max 3-SAT}&\textbf{Lem \ref{lem:minunweightedcyclekcycle}}\\
	
	\text{U. Radius} & m^{2\omega/3-o(1)}&\textbf{k-clique}&\textbf{Lem \ref{lem:unweightedundRadiusisMN}}\\
	\text{U. Radius} & m^{3/2-o(1)}&\textbf{Max 3-SAT}&\textbf{Lem \ref{lem:unweightedundRadiusisMN}}\\

	\text{U. Wiener Ind.} &   m^{2\omega/3-o(1)}&\textbf{k-clique}&\textbf{Lem \ref{lem:unweightedundWeinerMN}}\\
	\text{U. Wiener Ind.} & m^{3/2-o(1)}&\textbf{Max 3-SAT}&\textbf{Lem \ref{lem:unweightedundWeinerMN}}\\
	
	\text{U. APSP} & m^{2\omega/3-o(1)} &\textbf{k-clique}&\textbf{Cor \ref{cor:APSPhard}}\\
	\text{U. APSP} & m^{3/2-o(1)} &\textbf{Max 3-SAT}&\textbf{Cor \ref{cor:APSPhard}}\\
	\hline
	\end{array}
	$
	\centering
	\caption{Unweighted graph lower bounds. Our results are in bold. Upper bounds marked with $*$ are conjectured. U stands for undirected. Src stands for source. Shrt stands for shortest. Ind stands for index.}
	\label{table:unweighted}
\end{table*}

Our new lower bounds for sparse unweighted graph problems are summarized in Table \ref{table:unweighted}. Odd k-cycle is conjectured  to run in time $\tO(m^{(k+1)\omega / (2\omega +k +1)})$ by~\cite{yuster2004detecting}. The fastest algorithms for all other problems in the table run in time $O(\min\{mn, n^\omega\})$~\cite{itai1978finding,seidel1995all}.

\paragraph{Overview}
See Figure \ref{fig:reductionDiagram} for a depiction of our core reductions. 

In Sections \ref{sec:redFromHyperCliquetoHypercycle} to \ref{sec:shortestCycleHard} we cover the core reductions and show they are tight to the best known algorithms. 
The reduction from hyperclique to hypercycle is covered in Section \ref{sec:redFromHyperCliquetoHypercycle}.
The reduction from hypercycle to directed cycle in Section \ref{sec:hyperToCycle}. 
The algorithms for weighted minimum $k$-cycle which match the conditional lower bounds are discussed in Section \ref{sec:cycleUpperBounds}.
The reduction from minimum weight clique to shortest cycle is in Section \ref{sec:shortestCycleHard}.

In Sections \ref{sec:hypothesisDiscussion} to \ref{sec:MaxKSatHypercycle} we give justification for the hardness of the unweighted versions of these problems. In Section \ref{sec:hypothesisDiscussion} we discuss the hyper-clique hypothesis and give justification for it. In Section \ref{sec:genMatrixMult} we show that the generalized matrix product related to finding hypercliques in $k$-uniform hypergraphs can not be sped up with a Strassen like technique. In Section \ref{sec:MaxKSatHypercycle} we reduce Max-k-SAT to Tight Hypercycle. 

In Appendix \ref{sec:relatedwork} we discuss the prior work getting fast algorithms for the sparse graph problems we study. 
In Appendix \ref{sec:newreducs} we present the reductions from minimum $k$-cycle and minimum cycle to Radius, Weiner Index and APSP. In Appendix \ref{sec:generalCSPtoHyperclique} we reduce general CSP to the Hyperclique problem. Finally, in Appendix \ref{sec:nonIntgerDensities} we extend our lower bounds to make improved but non-matching lower bounds for graph densities between $n^{1+1/\ell}$ and $n^{1+1/(\ell+1)}$.

%% file: prelims.tex
In this section we define various notions that we will be using and prove some simple lemmas. 

\paragraph{Definitions and notation.}

Throughout this paper will be discussing problems indexed by $k$ and $\ell$. For example, $k$-cycle, $k$-clique, $(\ell, k)$-Hyperclique. We will treat the $k$ and $\ell$ values as being constant in these problems.  
%
	A \emph{hypergraph} $G=(V,E)$ is defined by its vertices $V$ and its hyperedges $E$ where each $e\in E$ is a subset of $V$. $G$ is a $k$-uniform hypergraph if all its hyperedges are of size $k$. 

Graphs are just $2$-uniform hypergraphs. Unless otherwise stated, the variables $m$ and $n$ will refer to the number of hyperedges and vertices of the hypergraph in question.
Unless otherwise stated, the graphs in this paper will be directed. Hypergraphs will not be directed. We will use node and vertex interchangeably.

	An \emph{$\ell$-hypercycle} in a $k$-uniform hypergraph is an ordered $\ell$-tuple of vertices $v_1,\ldots,v_\ell$ such that for every $i\in {1,\ldots,\ell}$, $(v_{i},v_{i+1},\ldots,v_{i+k-1})$ is a hyperedge (where the indices are mod $k$). 

We will be dealing with simple hypercycles, so that all $v_i$ are distinct. 
These types of hypercycles are known as tight hypercycles. We will omit the term tight for conciseness.

	An \emph{$\ell$-hyperclique} in a $k$-uniform hypergraph is a set of $\ell>k$ vertices $v_1,\ldots,v_\ell$ such that all subsets of $k$ of them $v_{i1},\ldots,v_{ik}$ form a hyperedge.

	A \emph{ $k$-circle-layered graph} is a $k$-partite directed graph $G$ where edges only exist between adjacent partitions. 
	More formally the vertices of $G$ can be partitioned into $k$ groups such that $V = V_1 \cup \ldots \cup V_k$ and $V_i \cap V_j = \emptyset$ if $i \ne j$. The only edges from a partition $V_i$  go to the partition $V_{i+1 \mod k}$. 

\paragraph{Hardness Hypotheses.}
We will state several hardness hypotheses that we will be using.

The first concerns the Min Weight $k$-Clique problem.
Min Weight $3$-Clique is known to be equivalent to APSP and other problems~\cite{williams2010subcubic}, and no truly subcubic algorithms are known for the problem. This issue extends to larger cliques: if the edge weights are large enough, no significantly faster algorithms than the brute-force algorithm are known. This motivates the following hypothesis used as the basis of hardness in prior work (see e.g.~\cite{BackursT16,abboud2014consequences}).
 

\begin{reminder}{Hypothesis~\ref{hyp:minkClique}} \emph{(Min Weight $k$-Clique Hypothesis).} There is a constant $c$ such that, on a Word-RAM with $O(\log n)$ bit words, finding a $k$-Clique of minimum total edge weight in an $n$-node graph with nonnegative integer edge weights bounded by $n^{ck}$ requires $n^{k-o(1)}$ time.  \end{reminder}

The exact weight version of the $k$-clique problem is at least as hard as Min Weight $k$-Clique~\cite{williams2013finding}, so that if the previous hypothesis is true, then so is the following one. For $k=3$, the Exact $3$-Clique problem is known to be at least as hard as both APSP and $3$-SUM, making the following hypothesis even more believable.

\begin{hypothesis}[Exact Weight $k$-Clique] There is a constant $c$ such that,
on a Word-RAM with $O(\log n)$ bit words, finding a $k$-Clique of total edge weight exactly $0$, in an $n$-node graph with integer edge weights bounded in $[-n^{ck},n^{ck}]$ requires $n^{k-o(1)}$ time.
\end{hypothesis}

Let $k\geq 3$ be an integer. The following hypothesis concerns the Max-$k$-SAT problem. The brute-force algorithm for Max-$k$-SAT on $n$ variables and $m$ clauses runs in $O(2^n m)$ time. There have been algorithmic improvements for the approximation of Max-$k$-SAT \cite{AsanoW02,AvidorBZ05,FeigeG95} and Max-2-SAT \cite{thesis}.
No $O(2^{(1-\eps)n})$ time algorithms are known for any $\eps>0$ for $k\geq 3$. Williams~\cite{Will05,thesis} showed that Max-$2$-SAT does have a faster algorithm running in $O(2^{\omega n/3}\cdot\poly (mn))$ time, however the algorithm used can not extend to Max $k$ SAT for $k>2$ (see the discussion in Section \ref{sec:genMatrixMult}).

\begin{hypothesis}[Max-$k$-SAT Hypothesis] On a Word-RAM with $O(\log n)$ bit words, given a $k$-CNF formula on $n$ variables, finding a Boolean assignment to the variables that satisfies a maximum number of clauses, requires $2^{n-o(n)}$ time.
\end{hypothesis}

The Max-$k$-SAT hypothesis implies the following hypothesis about hyperclique detection, as shown by Williams~\cite{thesis} for $k=3$ (see Appendix for the generalization for $k>3$). Williams~\cite{thesis} in fact showed that hyperclique detection solves even more difficult problems such as Satisfiability of Constraint Satisfaction Problems, the constraints of which are given by degree $k$ polynomials defining Boolean functions on the $n$ variables. Thus if the following hypothesis is false, then more complex MAX-CSP problems than MAX-$k$-SAT can be solved in $O(2^{(1-\eps)n})$ time for $\eps>0$.


\begin{reminder}{Hypothesis ~\ref{hyp:lkhypercliqef}} \emph{($(\ell,k)$-Hyperclique Hypothesis)}. Let $\ell>k>2$ be integers. On a Word-RAM with $O(\log n)$ bit words, finding an $\ell$-hyperclique in a $k$-uniform hypergraph on $n$ nodes requires $n^{\ell-o(1)}$ time. \end{reminder}

Abboud et al.~\cite{abboudunpub} have shown (using techniques from \cite{AbboudLW14}) that if the $(\ell,4)$-Hyperclique Hypothesis is false for some $\ell$, then the Exact Weight $\ell$-Clique Hypothesis is also false. Thus, the Hyperclique Hypothesis should be very believable even for $k=4$. The hypergraphs we are considering are dense ($m = \Theta(n^k)$). Hyperclique can be solved faster in hypergraphs where $m = o(n^{\ell/(\ell-1)})$ \cite{GaoIKW17}.

\paragraph{Simple $k$-Cycle Reductions.}

Note that throughout this paper we will use the fact that the $k$-cycle and $k$-clique problems we consider are as hard in $k$-partite graphs as they are in general graphs. Furthermore, the $k$-cycle problems we consider are as hard in $k$-circle layered graphs as they are in general graphs. Using the $k$-partite or $k$-circle layered versions often makes reductions more legible. 
%

$k$-cycle has different behavior when $k$ is even and odd. To get some results we will use a simple reduction from $k$ cycle to $k+1$ cycle. 

\begin{lemma} 
	Let $G=(V,E)$ be an $n$ node $m$ edge $k$-circle-layered graph. Suppose further that the edges have integer weights in $\{-M,\ldots,M\}$.
	Then in $O(m+n)$ time one can construct a $k+1$-partite directed graph $G'$ on $\leq 2n$ nodes and $\leq n+m$ edges with weights in $\{-M,\ldots,M\}$, so that $G'$ contains a directed $(k+1)$-cycle of weight $Y$ if and only if $G$ contains a directed $k$-cycle of weight $Y$.
	\label{lemma:evenodd}
\end{lemma}
\begin{proof}
	Take $V_2$, say, and split every node $v\in V_2$ into $v_0$ and $v_1$, placing a directed edge $(v_0, v_1)$ of weight $0$ and splitting the edges incident to $v$ among $v_0$ and $v_1$, so that $v_0$ gets all edges incoming from $V_1$ and $v_1$ gets all edges outgoing to $V_2$.
\end{proof}

An immediate corollary is:

\begin{corollary} 
	Suppose that there is a $T(n,m)$-time algorithm that can detect a (min-weight/ $0$-weight/ unweighted) $k+1$-cycle in a $k+1$-circle-layered directed $n$-node, $m$-edge graph, then there is a $O(m+n)+T(2n,m+n)$ time algorithm that can detect a (min-weight/ $0$-weight/ unweighted) $k$-cycle in a $k$-circle-layered $n$-node, $m$-edge directed graph.
\end{corollary} 

The following Lemma allows us to assume that all graphs that we are dealing with are circle-layered. 

\begin{lemma} 
	Suppose that a (min-weight/ $0$-weight/ unweighted) $k$-cycle can be detected in $T(n,m)$ time in a $k$-circle-layered directed graph where the edges have integer weights in $\{-W,\ldots,W\}$.
	Then in $\tilde{O}(k^k (m+n+T(m,n)))$ time 
	one can detect a(min-weight/ $0$-weight/ unweighted) $k$-cycle in a directed graph $G$ (not necessarily $k$-circle-layered) on $n$ nodes and $ m$ edges with weights in $\{-W,\ldots,W\}$.
	\label{lem:colorcode}
\end{lemma}
\begin{proof}
We use the method of color-coding~\cite{AlonYZ16}. We present the randomized version, but this can all be derandomized using $k$-perfect families of hash functions, resulting in roughly the same runtime.
Every node in the graph selects a color from $\{1,\ldots,k\}$ independently uniformly at random. 
We take the original graph and we only keep an edge $(u,v)$ if $c(v)=c(u)+1 \mod k$ and we remove edges that do not satisfy this condition. The created subgraph $G'$ is $k$-partite - there is a partition for each color, and by construction, the edges only go between adjacent colors, so that the graph is $k$-circle-layered.

Since $G'$ is a subgraph of $G$, if $G'$ has a $k$-cycle $C$, then $C$ is also a $k$-cycle in $G$. 
Suppose now that $G$ has a $k$-cycle $C=\{u_1,\ldots,u_k\}$. If for each $i$, $c(u_i)=i$, then $C$ is preserved in $G'$. Thus, $C$ is preserved with probability at least $1/k^k$, and repeating $O(k^k \log n)$ times, we will find $C$ whp.
  %
   %
\end{proof}

%% file: hypercliquetohypercycle.tex

In this section we will reduce the problem of finding an $\ell$-hyperclique in a $k$-uniform hypergraph to finding an $\ell$-hypercycle in a $\gamma(\ell,k)$-uniform hypergraph for some function $\gamma$ which is roughly $(k-1)\ell/k$.

By a color-coding argument we can assume that the hypergraph is $k$-partite- the vertex set is partitioned into $k$ parts $\{V_i\}$ so that no hyperedge contains two nodes in the same $V_i$. The color-coding approach reduces the hyperclique problem to $2^{O(k)}\log n$ instances of the $k$-partite hyperclique problem.  A simple randomized approach assigns each vertex a random color from $\{1,\ldots,k\}$, and then part $V_i$ includes the vertices colored $i$. One removes all hyperedges containing two vertices colored the same and argues that any particular $k$-hyperclique has all its vertices colored differently with probability $1/(2e)^k$. Thus $2^{O(k)}\log n$ instances of the $k$-partite hyperclique problem suffice with high probability. The approach can be derandomized with standard techniques.

In the following theorem an arc will refer to a valid partial list of nodes from a hyperclique or hypercycle. This usage is attempting to get across the intuition that a set of nodes in a hyperclique can be covered by a small number of overlapping sets if those sets are large enough. See Figure \ref{fig:coverhyperportion} for an image depiction. 

We will hence prove the following theorem:

\begin{theorem}
Let $G$ be a $k$-uniform hypergraph on $n$ vertices $V$, partitioned into $\ell$ parts $V_1,\ldots,V_\ell$. 
Let $\gamma=\ell-\lceil \ell/k\rceil +1$.
In $O(n^\gamma)$ time we can create a $\gamma$-uniform hypergraph $G'$ on the same node set $V$ as $G$, so that $G'$ contains an $\ell$-hypercycle if and only if $G$ contains an $\ell$-hyperclique with one node from each $V_i$.

If $G$ has weights on its hyperedges in the range $[-W,W]$, then one can also assign weights to the hyperedges of $G'$ so that a minimum weight $\ell$-hypercycle in $G'$ corresponds to a minimum weight $\ell$-hyperclique in $G$ and every edge in the hyperclique has weight between $[-\binom{\gamma}{k}W,\binom{\gamma}{k}W]$. Notably, $\binom{\gamma}{k}\leq O(\ell^k)$.

\label{thm:hypercycleHyperclique}
\end{theorem}

\begin{proof}
Consider the numbers $1,\ldots,\ell$ written in order around a circle and let $i_1<i_2<\ldots<i_k$ be any $k$ of them. 
We are interested in covering all these $k$ numbers by an arc of the circle. 
What is the least number of numbers from $1$ to $\ell$ an arc covers if it covers all the $i_j$?

It's not hard to see that the arc starts at one of the $i_j$, goes clockwise and ends at $i_{j-1}$ (indices mod $k$).
Let $s(j)$ be the number of numbers strictly between $i_{j-1}$ an $i_{j}$.
The number of numbers that the arc contains is thus $\ell-s(j)$, and that the best arc 
picks the $j$ that maximizes $s(j)$.

\begin{figure}[h]
	\centering
	\includegraphics[scale=0.5]{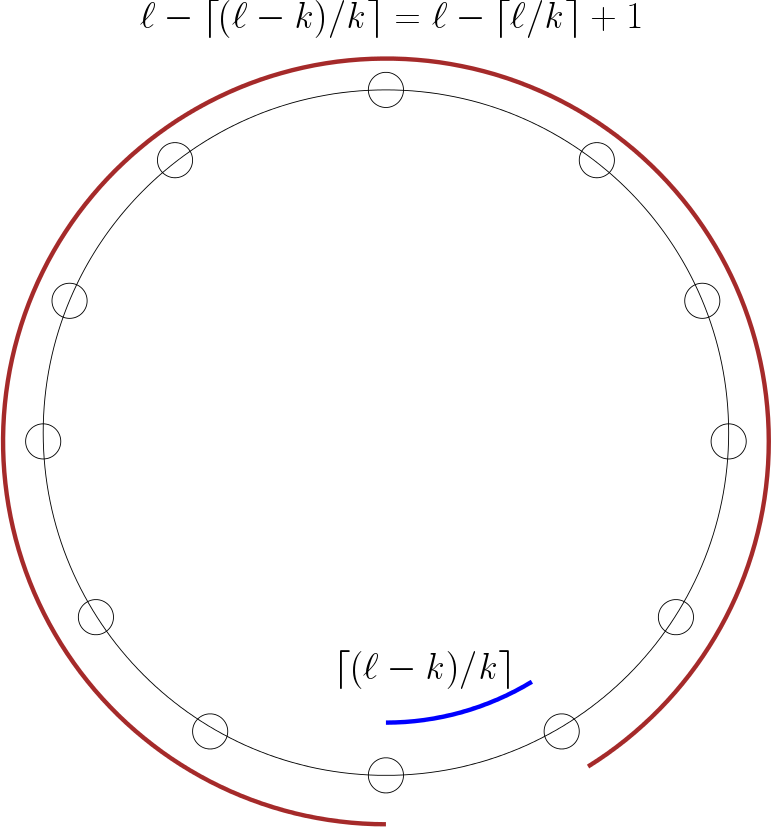}
	\caption{A depiction of why hypercycle needs sets of size $\ell-\lceil \ell/k\rceil +1$ to cover every choice of $k$ elements.}
	\label{fig:coverhyperportion}
\end{figure}

The sum $\sum_{j=1}^k s(j)$ equals $\ell-k$, and hence the maximum $s(j)$ is at least the average and is thus $\geq \lceil (\ell-k)/k\rceil$. Hence the best arc has at most $\ell-\lceil \ell/k\rceil+1$ numbers. See Figure \ref{fig:coverhyperportion}.

Now, let $G$ be the given $\ell$-partite $k$-uniform hypergraph in which we want to find an $\ell$-hyperclique. Let $V_1,\ldots,V_\ell$ be the vertex parts and let $E$ be the set of $k$-hyperedges. We will build a new hypergraph on the same set of nodes but with hyperedges of size $\gamma=\ell-\lceil \ell/k\rceil+1$ as follows.

Consider each $i\in [\ell]$ and every choice of nodes $u_i\in V_i, u_{i+1}\in V_{i+1},\ldots, u_{i+\gamma-1}\in V_{i+\gamma-1}$ call the set of chosen nodes $U$, i.e. nodes in $\gamma$ consecutive parts (mod $\ell$). We need only consider the sets of $\gamma$ consecutive parts because every subset of size $k$ will be contained in one of these sets, by our choice of $\gamma$. We add a hyperedge between the nodes in $U$ if every size $k$ subset of $U$ forms a hyperedge in $G$. That is, we create a big hyperedge in $G'$ if all the $k$-tuples contained in it form a hyperedge in $G$. The runtime to create $G'$ is $O(n^\gamma)$ as is the number of hyperedges created.  Clearly $G'$ is $\gamma$-uniform. 

Now suppose that $a_1\in V_1,\ldots,a_\ell\in V_\ell$ is an $\ell$-hyperclique in $G$. All the hyperedges $(a_i,\ldots,a_{i+\gamma-1})$ are present in $G'$ so $a_1,\ldots,a_\ell$ forms an $\ell$-hypercycle in $G'$.

Now suppose that $a_1\in V_1,\ldots,a_\ell$ is an $\ell$-hypercycle in $G'$. Consider $A=\{a_1,\ldots,a_\ell\}$ in $G$. We will show that it is an $\ell$-hyperclique. Let $a_{i_1},\ldots,a_{i_k}$ for $i_1<i_2<\ldots<i_k$ be any $k$ nodes of $A$. 

Let $t$ be the index that maximizes $s(i_t)$ as in the beginning of the proof. Then, $\{i_t,i_t+1,\ldots,i_{t-1}\}$ (which contains all $i_1,\ldots,i_k$) contains at most $\gamma$ nodes and is thus contained in $\{i_t,i_t+1,\ldots,i_{t+\gamma-1}\}$ which is a hyperedge in $G'$ since $a_1\in V_1,\ldots,a_\ell$ is an $\ell$-hypercycle in $G'$. However by the way we constructed the hyperedges, it must be that $(a_{i_1},\ldots,a_{i_k})$ is a hyperedge of $G$. Thus all $k$-tuples are hyperedges in $G$ and $\{a_1,\ldots,a_\ell\}$ is an $\ell$-hyperclique in $G$.

So far we have shown that we can construct a hypergraph so that the $\ell$-hypercliques in $G$ correspond to the $\ell$-hypercycles in $G'$. 
Suppose now that $G$ is a hypergraph with weights on its hyperedges. We will define weights for the hyperedges of $G'$ so that the weight of any $\ell$-hypercycle in $G'$ equals the weight of the $\ell$-hyperclique in $G$ that it corresponds to.
To achieve this, we will assign each hyperedge $y$ of $G$ to some hyperedges $E(y)$ of $G'$ and we will say that these hyperedges are responsible for $y$. Then we will set the weight of a hyperedge $e$ of $G'$ to be the sum of the weights of the hyperedges of $G'$ that it is responsible for. We will guarantee that for any hypercycle of $G'$, no two hyperedges in it are responsible for the same hyperedge of $G$, and that every hyperedge of the hyperclique that the hypercycle is representing is assigned to some of the hypercycle hyperedges.

Consider any hyperedge of $G$, $A=(a_{i_1},\ldots,a_{i_k})$ with $a_{i_j}\in V_{i_j}$. Let $i_t$ be the smallest index that maximizes $s(i_t)$. We assign $A$ to every hyperedge of $G'$ contained in $V_{i_t}\times V_{i_t+1}\times \ldots \times V_{i_{t+\gamma-1}}$ that intersects $V_{i_j}$ exactly at $a_{i_j}$. Then notice that any $\ell$-hypercycle that contains $a_{i_1},\ldots,a_{i_k}$ contains exactly one of these hyperedges, so that the weight of the hypercycle is exactly the weight of the hyperclique that it corresponds to. Since every hyperedge of $G'$ contains $\binom{\gamma}{k}$ hyperedges of $G$, the weights of the hyperedges lie in $[-W \binom{\gamma}{k},W \binom{\gamma}{k}]$.

%
%
%
\end{proof}


%% file: hypercycleDirectedcycle.tex
We have shown hardness for hypercycle from hyperclique. However, in order to get results on cycles in normal graphs we have to show that hypercycle can be solved efficiently with cycles in graphs. We do so below. 

\begin{lemma}
	Given an $n$-node $\lambda$ uniform hypergraph $H$ with nodes partitioned into $V_1,\ldots, V_k$ in which one wants to find a $k$-hypercycle $(v_1,v_2,\ldots,v_k)$ with $v_j\in V_j$ for each $j$, one can in $\tilde{O}(n^{\lambda})$ time create a $k$-circle-layered directed graph $G$ on $O(n^{\lambda-1})$ nodes and $O(n^{\lambda})$ edges, so that $H$ contains a $k$-hypercycle with one node in each partition $V_i$ if and only if $G$ contains a directed $k$-cycle.
	Moreover, if $H$ has integer weights on its edges bounded by $M$, then one can add integer edge weights to the edges of the graph $G$, bounded by $M$, so that the minimum weight $k$-cycle in $G$ has the same weight as the minimum weight $k$-hypercycle in $H$.
	
If $k$ is odd, the $G$ can be made undirected.	
	\label{lem:directedCycleHyperCycle}
\end{lemma}
\begin{proof}
Recall that a $k$-hypercycle in a $\lambda$-uniform hypergraph is formed by having a list of $k$ nodes $v_1,v_2,\ldots,v_k$ and having a hyperedge for all choices of $i\in[1,k]$ formed by the set $v_i,v_{i+1},\ldots,$ $ v_{i+\lambda-1}$ where we consider indices mod $k$.

We describe the construction of the directed graph $G$. It will be $k$-circle-layered with node parts $U_1,\ldots,U_k$.
For each $i\in \{1,\ldots, k\}$, we will add a node in part $U_i$ of $G$ for every choice of $\lambda-1$ nodes $v_i,\ldots,v_{i+\lambda-2}$ such that $v_i \in V_i$, $v_{i+1} \in V_{i+1}$ , $\ldots$, $v_{i+\lambda-2} \in V_{i+\lambda-2}$. This totals $n^{\lambda-1}$ nodes. Call this node $((v_i,\ldots,v_{i+\lambda-2}))$.

We will add a directed edge in $G$ between nodes $((v_i,\ldots,v_{i+\lambda-2}))$ and $((v'_{i+1},\ldots,v'_{i+\lambda-1}))$ if $v'_j = v_j$ for $j\in[i+1,i+\lambda-2]$ and $\{v_i,\ldots,v_{i+\lambda-2},v'_{i+\lambda-1}\}$ is a hyperedge in $H$. Assign this edge the weight of the hyperedge $\{v_i,\ldots,v_{i+\lambda-2},v'_{i+\lambda-1}\}$ in $H$. Every node in $G$ can connect to a maximum of $n$ other nodes giving us $|E|= O(n^{\lambda+1})$.

Now note that $G$ is a $k$-circle-layered graph. Further note that if a $k$-cycle exists in $G$ then each of its edges corresponds to a hyper-edge edge in $H$ and the set of vertices represented in the $k$-cycle in $G$ corresponds to a choice of $k$ nodes $v_1,v_2,\ldots,v_k$. Further, every edge of $G$ covers $\lambda$ adjacent vertices from $v_1,v_2,\ldots,v_k$.
	
We also note that if $k$ is odd, then the edges of $G$ can be made undirected: any $k$ cycle in $G$ must have a node from each $U_i$, as removing any $U_i$ from $G$ makes it bipartite, and no odd cycles can exist in a bipartite graph.
\end{proof}

We immediately obtain the following corollaries:

\begin{corollary}
	Let $\lambda = \ell - \ceil{\ell/k}+1$.
	Under the $(\ell,k)$-Hyperclique Hypothesis, min weight $\ell$-cycle in directed graphs (or in undirected graphs for $\ell$ odd) cannot be solved in $O(m^{\ell/\lambda -\eps})$ time for any $\eps>0$ for $m=\Theta(n^{1 + 1/(\lambda-1)})$ edge, $n$ node graphs.
\end{corollary}
\begin{proof}
	We start with a $k$-uniform hypergraph with $n_{old}$ nodes. 	
	The number of edges in the graph produced by Lemma 
	\ref{lem:directedCycleHyperCycle} when applied to this hypergraph is $(n_{old})^\lambda$. By the $(\ell,k)$-Hyperclique Hypothesis any algorithm to find a $(\ell,k)$-Hyperclique should take $(n_{old})^\ell$ time. Combining these facts we get a bound of $\Omega(m^{\ell/\lambda -o(1)})$.
	
	The number of nodes produced by Lemma 
	\ref{lem:directedCycleHyperCycle} is $n = (n_{old})^{\lambda-1}$, the number of edges is $m = (n_{old})^{\lambda}$. Thus, $m = n^{1+1/(\lambda-1)}$
\end{proof}

\begin{corollary}
	Let $\lambda = k - \ceil{k/2}+1$.
Under the Min Weight $k$-Clique Hypothesis, min weight $k$-cycle in directed graphs (or in undirected graphs for $k$ odd) cannot be solved in $O(nm^{\ceil{k/2}/\lambda - \eps})$ time for any $\eps>0$ for $m=\Theta(n^{1 + 1/(\lambda-1)})$ edge, $n$ node graphs.
\label{cor:cliqueToCycle}
\end{corollary}
\begin{proof}
The Min Weight $k$-Clique Hypothesis is equivalent to the Min Weight $(k,2)$-Hyperclique Hypothesis. We can plug in these numbers to get the result above. 
\end{proof}

When considering odd sizes of cliques and cycles, these results become show $O(mn)$ hardness for the cycle problems at certain densities. 

\begin{corollary}
	Under the Min Weight $(2k+1)$-Clique Hypothesis, min weight $(2k+1)$-cycle in directed or undirected graphs cannot be solved in $O(mn^{1-\eps})$ time for any $\eps>0$ for $m=\Theta(n^{1 + 1/k})$ edge, $n$ node graphs.
	\label{cor:cliqueToCycleODD}
\end{corollary}
\begin{proof}
	The Min Weight $(2k+1)$-Clique Hypothesis is equivalent to the Min Weight $((2k+1),2)$-Hyperclique Hypothesis. We can plug in these numbers to get the result above. 
	We then note that directed $(2k+1)$-cycle is solved by undirected $(2k+1)$-cycle because $(2k+1)$ is odd. 
\end{proof}

\begin{corollary}
	Under the Exact Weight $(2k+1)$-Clique Hypothesis, exact weight $(2k+1)$-cycle in directed and undirected graphs cannot be solved in $O(mn^{1-\eps})$ time for any $\eps>0$ for $m=\Theta(n^{1 + 1/k})$ edge, $n$ node graphs.
\end{corollary}
\begin{proof}
The Exact Weight $(2k+1)$-Clique Hypothesis is equivalent to the Exact Weight $((2k+1),2)$-Hyperclique Hypothesis. We can plug in these numbers to get the directed version of the above corollary. We then note that directed $(2k+1)$-cycle is solved by undirected $(2k+1)$-cycle because $(2k+1)$ is odd. 
\end{proof}

%% file: shortestCycleUB.tex
The reductions from $\ell$-hyperclique in $k$-uniform hypergraphs (through hypercycle) to directed $\ell$-cycle produces graphs on $O(n^{\gamma-1})$ nodes and $O(n^\gamma)$ edges where $\gamma=\ell-\lceil \ell/k\rceil +1$. 

For the special case of the reduction from Min Weight $\ell$-Clique ($k=2$), one obtains a graph on $O(n^{\lfloor\ell/2\rfloor +1})$ edges.
Suppose that $\ell$ is odd. The number of edges in the graph is $O(n^{(\ell+1)/2})$, and solving the Shortest $\ell$-Cycle problem in this graph in $O(n^{\ell-\eps})$ time for any $\eps>0$ would refute the Min Weight $\ell$-Clique Hypothesis.
We immediately obtain that Min Weight $\ell$-Cycle on $m$ edge graphs requires $m^{2-2/(\ell+1) - o(1)}$ time.

Using Lemma~\ref{lemma:evenodd}, we can also conclude that if $\ell$ is even, then solving Min Weight $\ell$-Cycle on $m$ edge graphs requires $m^{2-2/\ell - o(1)}$ time.

\begin{theorem}
Assuming the Min Weight $\ell$-Clique Hypothesis, on a Word RAM on $O(\log n)$ bit words, Min Weight $\ell$-Cycle on $m$ edge graphs requires $m^{2-2/\ell - o(1)}$ time if $\ell$ is even and $m^{2-2/(\ell+1) - o(1)}$ time if $\ell$ is odd.
\end{theorem}

The rest of this section will show that the above runtime can be achieved:

\begin{theorem}
Min Weight $\ell$-Cycle on $m$ edge graphs can be solved in $\tilde{O}(m^{2-2/\ell})$ time if $\ell$ is even and $\tilde{O}(m^{2-2/(\ell+1)})$ time if $\ell$ is odd.
\end{theorem}

The proof proceeds analogously to Alon, Yuster and Zwick's algorithm~\cite{alon1997finding} for $\ell$-Cycle in unweighted directed graphs. 
Let us review how their algorithm works and see how to modify it to handle weighted graphs.
First, pick a parameter $\Delta$ and take all $O(m/\Delta)$ nodes of degree $\geq \Delta$. Call the set of these nodes $H$. For every $v\in H$, Alon, Yuster and Zwick use an $\tilde{O}(m)$ time algorithm by Monien~\cite{moniencycle} to check whether there is an $\ell$-cycle going through $v$. If no $\ell$-cycle is found, they consider the subgraph with all nodes of $H$ removed and enumerate all $\lceil \ell/2 \rceil$-paths $X$ and all $\lfloor \ell/2\rfloor$-paths $Y$ in it. The number of paths in $X\cup Y$ is $\leq m\Delta^{\ceil{\ell/2}-1}$. 
Then one sorts $X$ and $Y$ in lexicographic order of the path endpoints and searches in linear time in $|X|+|Y|$ for a path in $X$ from $a$ to $b$ and a path in $Y$ from $b$ to $a$. To make sure that the cycle closed by these paths is simple, one can first start by color coding in two colors red and blue and let $X$ contain only paths with red internal nodes and $Y$ only paths with blue internal nodes, or one can just go through all paths that share the same end points. Either way, the total runtime is asymptotically $m^2/\Delta + m\Delta^{\ceil{\ell/2}-1}$, and setting $\Delta=m^{1/\ceil{\ell/2}}$ gives a runtime of $\tilde{O}(m^{2-1/(\lceil \ell/2\rceil )})$.

One can modify the algorithm to give a Shortest $\ell$-cycle in an edge-weighted graph, as follows. First, we replace Monien's algorithm with an algorithm that given a weighted graph and a source $s$ can in $\tilde{O}(m)$ time determine a shortest $\ell$-cycle $C$ containing $s$. 
To this end, we use color-coding: we give every node of $G$ a random color from $1$ to $\ell$ and note that with probability at least $1/\ell^\ell$, the $i$th node of $C$ is colored $i$, for all $i$, whp; as $s$ is the first node of $C$, we can assume that $s$ is colored $1$. As usual, this can be derandomized using $\ell$-perfect hash families. Now, in $G$, only keep the edges $(u,v)$ such that $c(v)=c(u)+1$ (not mod $\ell$, so there are no edges between nodes colored $\ell$ and nodes colored $1$). This makes the obtained subgraph $G'$ $\ell$-partite and acyclic. Now, run Dijkstra's algorithm from $s$, computing the distances $d(s,v)$ for each $v\in V$. Then for every in-neighbor $u$ of $s$ in $G$ colored $\ell$, compute $d(s,u)+w(u,s)$ and take the minimum of these, $W=\min_{u} d(s,u)+w(u,s)$. If the nodes of $C$ are colored properly (the $i$th node is colored $i$), then $W$ is the weight of the shortest $\ell$-cycle through $s$ since the shortest path from $s$ to any $u$ colored $k$, if the distance is finite, must have $\ell$ nodes colored from $1$ to $\ell$.
Dijkstra's algorithm runs in $\tilde{O}(m)$ time, and one would want to repeat $O(\ell^\ell \log n)$ times to get the correct answer with high probability (the same cost is obtained in the derandomization).

Now that we have a counterpart of Monien's algorithm, let's see how to handle the case when the shortest $k$-cycle in the graph only contains nodes of low degree.
Similar to the original algorithm, we again compute the set of paths $X$ and $Y$, but we only consider shortest paths together with their weights. Then one is looking for two paths (one between $a$ and $b$ and the other between $b$ and $a$) so that their sum of weights is minimized. This can also be found in linear time in $|X|$ and $|Y|$ when they are sorted by end points $(a,b)$ and by weight. 
The total runtime is again $\tilde{O}(m^{2-1/(\lceil \ell/2\rceil )})$.

%
%
%
%

%% file: shortestkvsShort.tex
\begin{theorem}
	If Shortest cycle in an $N$ node, $M$ edge directed graph can be solved in $T(N,M)$ time, then the Minimum Weight $k$-cycle in an $n$ node, $m$ edge directed graph is solvable in $\tO(T(n,m))$ time.
	\label{lem:mincyclekcycle}
\end{theorem}
\begin{proof}
	Let the weights of the $k$-cycle instance range between $-W$ and $W$.
	%
	Use Lemma~\ref{lem:colorcode} to reduce the Min Weight $k$-cycle problem to one in a $k$-circle-layered graph $G$ with partitions $A_1,A_2,\ldots,A_k$.
	Add the value $4W$ to each edge, which adds $4kW$ to the value of every $k$-cycle.  Every cycle in a directed $k$-circle-layered graph is a $ck$-cycle when $c$ is a positive integer since every cycle must go around the graph circle some number of times. Due to the added weight $4W$, the Shortest cycle in the new graph will minimize the number of edges:
	Any $ck$-cycle $C$ for $c\geq 2$ will have weight $\geq 4ckW + w(C)\geq 3ckW\geq 6kW$, where $w(C)\geq -Wck$ is the weight of $C$ in $G$.
	The weight of a $k$ cycle, however is at most $4kW+kW=5kW<6kW$. Thus, the weight of the Shortest Cycle in the new graph is exactly the weight of the Min Weight $k$-Cycle in $G$, plus $4kW$, and the Shortest Cycle will exactly correspond to the Min Weight $k$-Cycle in $G$.	
\end{proof}

\begin{lemma}
	If Shortest Cycle can be solved in $T(n,m)$ time in an $n$-node, $m$-edge directed unweighted graph, then $k$-cycle in a directed unweighted $n$-node, $m$-edge graph is solvable in $\tO(T(n,m))$ time.
	\label{lem:minunweightedcyclekcycle}
\end{lemma}
\begin{proof}
The proof is similar but simpler than that of Theorem \ref{lem:mincyclekcycle}. We first reduce to $k$-cycle in a $k$-circle-layered graph, and then just find the Shortest Cycle in it. Since the graph obtained is directed and $k$-circle-layered, if it contains a $k$-cycle, then that cycle is its shortest cycle. 
\end{proof}

\begin{corollary}
	If Min Weight $(2L+1)$-clique requires $n^{2L+1-o(1)}$ time, then Shortest Cycle in directed weighted graphs requires $mn^{1-o(1)}$ time whenever $m = \Theta(n^{1+1/L})$.
	
	Directed Shortest Cycle in unweighted graphs requires $m^{3/2-o(1)}$ time under the Max $3$ SAT Hypothesis, $m^{4/3-o(1)}$ time under the Exact Weight $K$ Clique Hypothesis, and $m^{2\omega/3-o(1)}$ time under the $K$-Clique Hypothesis.
\end{corollary}
\begin{proof}
The first statement follows immediately from Lemma~\ref{lem:mincyclekcycle} and Corollary~\ref{cor:cliqueToCycle}.
We will focus on the second part of the corollary.

The reduction in Corollary~\ref{lem:hyperSATcycle} from Max $3$ SAT on $N$ variables to $\ell$-cycle (for any $\ell>3$) produces a $\left( n^{\ell-\lceil \ell/3\rceil} \right)$-node, $\left(n^{\ell-\lceil \ell/3\rceil+1}\right)$-edge graph (for $n=2^{N/\ell}$, so that solving $\ell$-cycle in it in $O(n^{\ell-\eps})$ time for any $\eps>0$, then the Max $3$ SAT Hypothesis is false.
Now suppose that Shortest cycle in a directed graph can be solved in $O(m^{3/2-\eps})$ time for some $\eps>0$.
Set $\ell$ to be any integer greater than $3/\eps$ and divisible by $3$.
Consider the $\ell$-cycle problem in $n^{\ell-\lceil \ell/3\rceil+1}$-edge graphs obtained via the reduction from Max $3$ SAT.
Reduce it to Shortest Cycle as in Lemma~\ref{lem:mincyclekcycle}. As $\ell$ is divisible by $3$, the number of edges in consideration is $O(n^{2\ell/3+1})$.
Then, applying the $O(m^{3/2-\eps})$ time algorithm, we can solve the $\ell$-cycle instance in $O(n^{(2\ell/3 + 1) (3/2-\eps)})$ time. As we set $\ell\geq 3/\eps$, the exponent in the running time is $\ell +3/2-\eps(2\ell/3+1)\leq \ell  +3/2-\eps(2/\eps+1)\leq \ell - 1/2 -\eps$, and hence we obtain a faster algorithm for $\ell$-cycle and contradict the Max $3$-SAT hypothesis.

A similar argument applies to show that $m^{4/3-o(1)}$ time is needed under the Exact Weight $K$ Clique Hypothesis, and $m^{2\omega/3-o(1)}$ time is needed under the $K$-Clique Hypothesis.
\end{proof}

%% file: hyperdiscuss.tex
\label{sec:discuss}
In this section we discuss why the $(\ell,k)$-Hyperclique hypothesis is believable.

First, when $k>2$, the fastest algorithms for the $\ell$-hyperclique problem run in $n^{\ell-o(1)}$ time, and this is not for lack of trying. Many researchers~\cite{researchers} have attempted to design a faster algorithm, for instance by mimicking the matrix multiplication approach for $k$-Clique. However, in doing this, one needs to design a nontrivial algorithm for a generalized version of matrix multiplication. Unfortunately, in Section \ref{sec:genMatrixMult}, we show that the rank and even the border rank of the tensor associated with this generalized product is as large as possible, thus ruling out the arithmetic circuit approach for the problem. Thus, if a faster than $n^\ell$ algorithm exists for $k$-uniform hypergraphs with $k>2$, then it must use radically different techniques than the Strassen-like approach to regular matrix multiplication.

Another reason to believe the Hyperclique hypothesis is due to its relationship to Maximum Constraint Satisfaction Problems (CSPs).
R. Williams~\cite{thesis} showed that Max-$3$-SAT can be reduced to finding a $4$-Hyperclique in a $3$-uniform hypergraph, so that if the latter can be solved in $O(n^{4-\eps})$ time for $n$ node graphs and $\eps>0$, then Max-$3$-SAT can be solved in $O(2^{(1-\delta)n})$ time for formulas on $n$ variables.

Max-$3$-SAT has long resisted attempts to improve upon the brute-force $2^{n}$ runtime. Recent results (e.g. \cite{AlmanCW16}) obtained $2^{n-o(n)}$ time improvements, but there is still no $O((2-\eps)^n)$ time algorithm.
 Generalizing the reduction from ~\cite{thesis} (see Section~\ref{sec:MaxKSatHypercycle}), one can reduce Max-$k$-SAT to $\ell$-hyperclique in a $k$-uniform hypergraph for any $\ell>k$, so that if the latter problem can be solved in $O(n^{\ell-\eps})$ time for $n$ node graphs and $\eps>0$, then Max-$k$-SAT can be solved in $O(2^{(1-\delta)n})$ time for formulas on $n$ variables.
In fact, R. Williams~\cite{thesis} showed that even harder Constraint Satisfaction Problems (CSPs) can be reduced to hyperclique. CSPs where the constraints are degree $3$ polynomials representing Boolean functions over the $n$ variables. In Section~\ref{sec:GeneralCSP} we generalize this to CSPs where the constraints are degree $k$ polynomials. Such CSPs include Max-$k$-SAT and also include some CSPs with constraints involving more than $k$ variables. In any case, the $(\ell,k)$-Hypothesis captures the difficulty of this very general class of CSPs.

Another reason to believe the Hypothesis is due to its relationship to the Exact Weight $k$-Clique Conjecture~\cite{williams2013finding} which states that finding a $k$-Clique of total edge weight exactly $0$ in an $n$ node graph with large integer weights  requires $n^{k-o(1)}$ time. The Exact Weight $k$-Clique conjecture is implied by the Min Weight $k$-Clique conjecture, so it is at least as believable. Furthermore, for the special case $k=3$, both $3$SUM and APSP can be reduced to Exact Weight $3$-Clique, so that a truly subcubic algorithm for the latter problem would refute both the APSP and the $3$SUM conjectures~\cite{patrascu2010towards,williams2013finding,williams2010subcubic}. Exact Weight $k$ Clique is thus a very difficult problem. Recent work by Abboud et al.~\cite{abboudunpub} shows how to use the techniques in \cite{AbboudLW14} to reduce the Exact Weight $k$-Clique problem to (unweighted) $k$-Clique in a $4$-uniform hypergraph. Thus, if one believes the Exact Weight $k$-Clique conjecture, then one should definitely believe the $(\ell,4)$-Hyperclique Hypothesis. (A generalization of this approach also shows that Exact Weight $\ell$-Hyperclique in a $k$-uniform hypergraph can be tightly reduced to (unweighted) $\ell$-hyperclique in a $2k$-uniform hypergraph.)

We note that the hypothesis concerns dense hypergraphs. For hyperclique in sparse hypergraphs, faster algorithms are known: the results of Gao et al.~\cite{GaoIKW17} imply that an $\ell$-hyperclique in an $m$-hyperedge, $n$-node $k$-uniform hypergraph (for $\ell>k$) can be solved in $m^{\ell-1}/2^{\Theta(\sqrt{\log m})}$.

%% file: nogeneralizedmmult.tex
The fastest known algorithm for $\ell$-clique reduces $\ell$-clique to triangle detection in a graph and then uses matrix multiplication to find a triangle \cite{NesetrilPoljak}. One might ask, is there a similar approach to finding an $\ell$-hyperclique in a $k$-uniform hypergraph faster than $O(n^\ell)$ time?

The first step would be to reduce $\ell$-hyperclique problem in a $k$-uniform hypergraph to $k+1$-hyperclique in a $k$-uniform hypergraph. This step works fine: Assume for simplicity that $\ell$ is divisible by $(k+1)$ so that $\ell=c(k+1)$. We will build a new graph $G'$. Take all $c$-tuples of vertices of $G$ and create a vertex in $G'$ corresponding to the tuple if it forms an $c$-hyperclique in $G$ (if $c<k$, any $c$-tuple is a hyperclique, and if $c\geq k$, it is a hyperclique if all of its $k$-subsets are hyperedges).
For every choice of $k$ distinct $c$-tuples, create
 a hyperedge in $G'$ on them if every choice of $k$ nodes from their union forms a hyperedge in $G$. Now, $k+1$-hypercliques of $G'$ correspond to $\ell$-hypercliques of $G$. $G'$ is formed in $O(n^{ck})$ time and has $O(n^c)$ nodes. Hence if a $(k+1)$-hyperclique in a $k$-uniform hypergraph on $N$ nodes can be found in $O(N^{(k+1)-\eps})$ time for some $\eps>0$, then an $\ell$-hyperclique in a $k$-uniform hypergraph on $n$ nodes can be found in $O(n^{ck}+n^{c(k+1)-\eps c})=O(n^{\ell-\delta})$ time for $\delta=\min\{\ell/(k+1),c\eps\} >0$.

Thus it suffices to just find $k+1$-hypercliques in $k$-uniform hypergraphs. Following the approach for finding triangles (the case $k=2$), we want to define a suitable matrix product.

In the matrix multiplication problem we are given two matrices and we are asked to compute a third.
Matrices are just tensors of order $2$. The new product we will define is for tensors of order $k$. We will call these $k$-tensors for brevity.
The natural generalization of matrix multiplication for $k$-tensors of dimensions $n\times \ldots \times n$ ($k$ times) is as follows. 

\begin{definition}[$k$-wise matrix product]
Given  $k$ $k$-tensors of dimensions $n\times \ldots \times n$, $A^1,\ldots, A^k$, compute 
the $k$-tensor $C$ given by
\begin{align*}
C[i_1,\ldots,i_k] = &\\
&\sum_{\ell\in [n]} A^1[i_1,\cdots,i_{k-1},\ell]\cdot A^2[i_2,\cdots,i_{k-1},\ell,i_k]\cdots \\
&~~~~~~~~~~\cdots A^k[\ell,i_k, i_1,\cdots, i_{k-2}].
\end{align*}
\end{definition}

The special case of $k=3$ was defined in 1980 by Mesner et al.~\cite{MesnerTriproduct}:
Given three $3$-tensors $A^1,A^2,A^3$ with indices in $[n]\times [n]\times [n]$
compute the product $C$ defined as 
$C[i,j,k]=\sum_{\ell\in [n]} A^1[i,j,\ell]\cdot A^2[j,\ell,k]\cdot A^3[\ell,k,i]$. 
The more general definition as above was defined later by~\cite{Gnang2011} and its properties have been studied within algebra and combinatorics, e.g.~\cite{gnangarxiv}.

Now, if one can compute the $k$-wise matrix product in $T(n)$ time, then one can also find a $k+1$-hyperclique in a $k$-uniform hypergraph in the same time: define $A$ to be the adjacency tensor of the hypergraph -- it is of order $k$ and has a $1$ for every $k$-tuple that forms a hyperedge; if the $k$-wise product of $k$ copies of $A$ has a nonzero for some $k$-tuple that is also a hyperedge, then the hypergraph contains a $k+1$-hyperclique.

Now the question is:
{\em ``Is there an $O(n^{k-\eps})$ time algorithm for $(k-1)$-wise matrix product for $k\geq 4$ and $\eps>0$?''}

A priori, it seems quite plausible that a faster than $n^k$-time algorithm exists for generalized matrix product, as all of the techniques developed for matrix multiplication carry over to the general case: tensor products, rank, border rank etc.
The only thing missing is a suitable base case algorithm. Many researchers have searched for such an algorithm (e.g. 
\cite{researchers}). 
Williams asked in \cite{rmathoverflow} whether an $O(n^{4-\eps})$ time algorithm exists for $k=3$ and $\eps>0$, and this question has remained unanswered.

We will show that there is {\em no} smaller arithmetic circuit than the trivial one for generalized matrix product. In other words, there can be no analogue of the fast algorithms for clique that carry over to hyperclique.

\begin{theorem}
For every $k\geq 3$, the border rank of the tensor of $k$-wise matrix multiplication is exactly $n^{k+1}$.\label{thm:tenrank}
\end{theorem}

\begin{proof}[Sketch.]
We give a sketch of the proof for $k=3$. The general case is similar but requires more indices.
Consider the tensor $t$ which is such that $t_{i,j,k,j',k',l',k'',l'',i'', l''', i''', j'''}$ is $1$ whenever $i=i''=i''', j=j'=j''', k=k'=k''$ and $0$ otherwise. Let $t$ be in $A^*\otimes B^*\otimes C^*\otimes D^*$. Define the flattening of $t$ which is a linear transformation $T:~A^*\otimes B^*\rightarrow C\otimes D$ that takes $a(i,j,k)^*\otimes b(j',k',l)^*$ to $c(k,l,i)\otimes d(l,i,j)$ if $j=j',k=k'$ and to $0$ otherwise. The vectors $c(k,l,i)\otimes d(l,i,j)$ are $n^4$ independent vectors that span the image of $T$, and hence the map has rank $n^4$. Since $T$ is a linear transformation it must also have border rank $n^4$. Because $T$ is a flattening of $t$, $t$ must have border rank at least $n^4$.
\end{proof}

Theorem~\ref{thm:tenrank} completely rules out the arithmetic circuit approach for $k>2$. Nevertheless, there might still be an $O(n^{k+1-\eps})$ time algorithm (at least for the Boolean version of the product) that uses operations beyond $+,\cdot, /$, for instance by working at the bit level. 
Obtaining such an algorithm is quite a challenging but very interesting endeavor.

%% file: hypergraphsmaxSAT.tex

We are dealing with graphs and SAT instances. For ease of reading we will capitalize the variables associated with graphs (i.e. $M= |E|$  and $N= |V|$) and leave the variables associated with SAT instances in lower case (i.e. $m$ for number of clauses and $n$ for number of variables). The current fastest known algorithms for Max-$k$-SAT come from Alman, Chan and Williams~\cite{AlmanCW16}.


\begin{lemma}
	For all integers $L \geq 3$ if we can detect L-hyper cliques in a $ k$ uniform hypergraph, $G$, with $N$ nodes in time $T(N)$ then we can solve max-k-SAT in  $\tO( (mn^k)^{\binom{L}{k}}T(L2^{n/L})+  mn^k2^{nk/L})$ time.
	\label{lem:hyperSATclique}
\end{lemma}
\begin{proof}
Let $n$ be the number of variables and $m$ the number of clauses in our max-3-SAT instance.

Notice that the clauses of $k$-SAT can be represented by a $k$ degree polynomial. Each clause only depends on $k$ variables and $x_i^2 = x_i$. 

So we can apply Theorem \ref{thm:maxHypercliqueMaxCSP}. 

If the unweighted $L$-hyperclique problem on a $k$-uniform hypergraph 
can be solved in time $T(n)$ then the maximal degree-$k$-CSP problem can be solved in $O((mn^{k})^{\binom{L}{k}}T(2^{n/L})+mn^k2^{kn/L})$ time. When $L$ and $k$ are constants

\end{proof}

\begin{corollary}
For all integers $L\geq 3$ if we can detect a tight L-hyper cycle in a $d = L -\ceil{L/3} +1  $ uniform  hypergraph, $G$, with  $N = L 2^{n/L}$ nodes in time $T(N)$ then we can solve max-3-SAT in  $\tO( m^{\binom{L}{3}}T(2^{n/L})+2^{n(d/L)}+  m^{\binom{L}{3}}  2^{3n/L})$ time.
		\label{lem:hyperSATcycle}
\end{corollary}
\begin{proof}
 Take Lemma \ref{lem:hyperSATclique} in the case of $k=3$ to get a reduction from max-3-SAT to $ m^{\binom{L}{k}}T(N)$ instances of L-hyperclique in a $k$ regular graph. let $d = L -\ceil{L/3} +1  $. We then use Theorem \ref{thm:hypercycleHyperclique} to generate a d-uniform hypergraph $G'$ on $N = L 2^{n/L}$ nodes such that $G$ has a L-hypercycle only if the associated L-hyperclique instance has a clique. The overhead of this reduction is $O(2^{n(d/L)} + m^{\binom{L}{3}}  2^{3n/L})$ time to create these graphs. It takes $m^{\binom{L}{3}}T(N)$ time to run L-hypercycle on the $m^{\binom{L}{3}}$ instances. Giving us a total time of $\tO( m^{\binom{L}{3}}T(2^{n/L})+2^{n(d/L)}+  m^{\binom{L}{3}}  2^{3n/L})$.
\end{proof}

\begin{corollary}
	If directed L-cycle can be solved in a graph $G$ with $N=2^{n(d-1)/L}$ nodes and $M=2^{nd/L}$, where $d = L -\ceil{L/3} +1 $, in time $T(N,M)$ then max-3-SAT can be solved in time $\tO(m^{\binom{L}{3}} T(N,M)+ M + mn^32^{3n/L})$.
	\label{cor:dircycSAT}
\end{corollary}
\begin{proof}
		If directed L-cycle can be solved in a graph $G$ with $N=2^{n(d-1)/L}$ nodes and $M=2^{nd/L}$  in time $T(N,M)$ then by Lemma \ref{lem:directedCycleHyperCycle} we can solve tight L-hyper cycle in a $d$ regular graph in time $O(T(N,M)+M)$.
		
		If we could detect a tight L-cycle  in a $d $ regular  hypergraph, $G$, with  $N = L 2^{n(d-1)/L}$ nodes in time $O(T(N,M)+M)$ then we can solve max-3-SAT in  $\tO(m^{\binom{L}{3}}T(N,M)+m^{\binom{L}{3}} M+  2^{3n/L})$ time.
		
\end{proof}

\begin{theorem}
		If directed L-cycle, for $L>3$, can be solved in a graph $G$ in time $O(M^{c_L-\epsilon})$, where $c_L = L/d =L/( L -\ceil{L/3} +1 ) $ then max-3-SAT can be solved in time $\tO( 2^{(1-\epsilon')n})$.
		\label{thm:dirCycleMaxSat}
\end{theorem}
\begin{proof}
	We can plug in the running time $O(M^{c_L-\epsilon})$ into Corollary \ref{cor:dircycSAT}. On a graph with $N=2^{n(d-1)/L}$ nodes and $M=2^{nd/L}$ this running time gives 
	$$T(N,M)= 2^{(c_L-\epsilon)nd/L}= 2^{(L/d-\epsilon)nd/L} = 2^{(1 -\epsilon d/L)n} .$$ 
	We get that max 3-SAT is solved in time $\tO( (mn^k)^{\binom{L}{k}}2^{(1-\epsilon d/L)n} +m^{\binom{L}{k}} 2^{nd/L}+ 2^{n/L})$. Note that $d/L<1$ for $L>3$ and $d$ and $L$ are constants. Further note that $(mn^k)^{\binom{L}{k}}$ is a polynomial factor. Thus, max 3-SAT is solved in  $\tO( 2^{(1-\epsilon')n})$.
\end{proof}

This theorem gives us a lower bound for detecting L-cycles for large constant values of $L$ of $\tOmg(M^{3/2})$

\begin{corollary}
	If we can solve directed L-cycle in time $O( M^{3/2-\epsilon}) $ for all constant $L$  then we can solve max-3-SAT in time $O(2^{1-\epsilon'}n)$.
	\label{cor:1.5LBforarbLargeCycle}
\end{corollary}
\begin{proof}
	
	We will use Theorem \ref{thm:dirCycleMaxSat} and set $L$ large enough that $c_L> 3/2-\epsilon$.
	
	As $L /( L -\ceil{L/3} +1 )> L/(2L/3+1) = 3/2 - 9/(4L+6)$. If $L >(9/\epsilon -6)/4$ then $c_L > 3/2-\epsilon$. We can now invoke Theorem \ref{thm:dirCycleMaxSat}. If we can solve  L-cycle in time $O( M^{3/2-\epsilon})$ then we can L-cycle in time $O( M^{c_L-\delta})$ for some $\delta>0$. Thus we can solve max-3-SAT in time $O(2^{(1-\epsilon')n})$.
\end{proof}

What if you care about specifically solving directed cycle for a particular constant $L$? Our result will garner improvements over the previous result of $\tOmg(M^{4/3})$ for many small cycle lengths. Notably for 7-cycle we get the bound of $\tOmg(M^{7/5})$ (and $7/5>4/3$). 
\begin{corollary}
	If we can solve directed 7-cycle in a graph with $M = N^{5/4}$ in time $O( M^{7/5-\epsilon}) $  then we can solve max-3-SAT in time $O(2^{(1-\epsilon')n})$.
	\label{cor:7cycleLB}
\end{corollary}
\begin{proof}
	Once again we will use Theorem \ref{thm:dirCycleMaxSat} and we will plug in $L=7$ note that $c_L = 7/5$ giving us the bound. 
\end{proof}

%% file: relatedwork.tex

In this section we provide some of the prior algorithmic work on the various problems that we study. We limit our scope to algorithms that compute the problems below exactly. We note that faster approximation algorithms are known for many of them.

\paragraph{$k$-Cycle.}
The complexity of finding a $k$-Cycle in an $n$ node, $m$ edge graph depends heavily on whether the graph is directed and whether $k$ is even or odd.
In directed graphs, the fastest algorithm for finding $k$-cycles run in $\tilde{O}(n^\omega)$ time~\cite{AlonYZ16}. For sparse graphs, \cite{alon1997finding} obtained faster algorithms using matrix multiplicatin. When $k$ is odd, the $k$-cycle problem in directed graphs is equivalent to the $k$-cycle problem in undirected graphs (see e.g. \cite{vthesis}). When $k$ is even, however, the problem is much easier in undirected graphs. Yuster and Zwick~\cite{evenfaster} showed that a $k$-cycle in an undirected graph can be found in $\tilde{O}(n^2)$ time for any even constant $k$. Dahlgaard et al.~\cite{DahlgaardKS17} recently extended this result, giving an $\tilde{O}(m^{2k/(k+2)})$ time algorithm for all even $k$.

\paragraph{Shortest Cycle.}
The Shortest Cycle problem in weighted graphs was recently shown to be solvable in $O(mn)$ time by 
Orlin and Sede{\~{n}}o{-}Noda \cite{OrlinS17}. In directed or undirected graphs with small integer weights in the interval $\{1,\ldots, M\}$, Roditty and Vassilevska Williams~\cite{roditty2011minimum} showed how to find a shortest cycle in $\tilde{O}(Mn^\omega)$ time, generalizing a result by Itai and Rodeh~\cite{itai1978finding} that obtained $O(n^\omega)$ for unweighted graphs.

\paragraph{APSP.}
In directed or undirected weighted graphs, the fastest algorithm for APSP runs in $n^3/2^{\Theta(\sqrt{\log n})}$ time~\cite{williams2014faster} for dense graphs. For APSP in sparse the fastest algorithm is $O(mn+n^2\log\log n)$ time~\cite{pettie2002faster} for directed graphs, and $O(mn \lg(\alpha(m,n)))$ in undirected graphs~\cite{PettieR05}. 

For undirected graphs with integer weights in $\{-M,\ldots,M\}$, Shoshan and Zwick~\cite{ShoshanZ99} provided an $\tilde{O}(Mn^\omega)$ time algorithm, generalizing Seidel's $O(n^\omega \log n)$ time algorithm for unweighted undirected graphs \cite{seidel1995all}.
Zwick~\cite{zwickbridge} obtained an $O(M^{0.681}n^{2.532})$ time algorithm for APSP in directed graphs with integer weights in $\{-M,\ldots,M\}$; his algorithm for $M=1$ gives the best known running time for APSP in directed unweighted graphs. Chan~\cite{chan2012all} obtained $O(mn \log\log n /\log n)$ time for sparse undirected graphs, while for APSP in sparse directed graphs the $O(mn)$ BFS-based algorithm is currently the fastest.

\paragraph{Radius.}
The fastest known algorithm for Radius in directed or undirected graphs with large edge weights just uses the best known APSP algorithms above. 
For graphs with integer edge weights in $\{-M,\ldots,M\}$, Cygan et al.~\cite{CyganGS15} provide an $\tilde{O}(Mn^\omega)$ time algorithm; when setting $M=1$ one gets the fastest known runtime of $\tilde{O}(n^\omega)$ for unweighted graphs.

\paragraph{Wiener Index.} The best algorithms for Wiener index compute APSP and sum up all the pairwise distances. (See Mohar and Pisanski~\cite{wienerindex} for versions of the Floyd-Warshall and other algorithms for Wiener index.)

\paragraph{Replacement Paths (RP) and Second Shortest Path (SSP).} 
RP and SSP for undirected
graphs can be solved very efficiently: Malik et al. ~\cite{malik1989k} presented an $\tilde{O}(m)$ time algorithm. 
Nardelli et al.~\cite{NardelliPW03} improved this runtime to $O(m\alpha(n))$ in the word-RAM model.
 The best runtime for both problems
in directed graphs with arbitrary edge weights is $O(mn +
n^2 \log \log n)$ by
Gotthilf and Lewenstein~\cite{gotthilf2009improved}.
For dense weighted graphs, the best algorithm use the best algorithms for APSP.
Vassilevska Williams~\cite{Williams11replace} obtained an $\tilde{O}(Mn^\omega)$ time algorithm for RP and SSP for directed graphs with edge weights in $\{-M,\ldots,M\}$.
For unweighted directed graphs, Roditty and
Zwick~\cite{RodittyZ12} gave a randomized combinatorial algorithm
which computes replacement paths $\tilde{O}(m\sqrt n)$
time.

\paragraph{Betweenness Centrality.}
The notion of betweenness centrality was introduced
by Freeman~\cite{freeman} in the context of social networks, and
since then became one of the most important graph centrality
measures in the applications. Brandes'
algorithm~\cite{brandes} computes the betweenness centrality of
all nodes in time $O(mn + n^2 \log n)$ using
a counting variant of Dijkstras algorithm. Similar to other papers in the area, \cite{brandes} neglects
the bit complexity of the counters storing the number
of pairwise shortest paths. This is reasonable in practice
since the maximum number $N$ of alternative shortest
paths between two nodes tends to be small. By also considering $N$, the running
time grows by a factor $O(\log N) = O(n \log n)$.

%% file: radius.tex
\subsection{Undirected k-cycle reduces to radius}

First we will demonstrate a reduction from $k$-cycle to Radius in the undirected weighted case. Radius will solve $k$-cycle at any graph density and thus we will get a lower bound of $mn^{1-o(1)}$ for all densities $m = n^{1+1/L}$.

To make the reduction more readable we will first prove a lemma that gives a reduction from minimum $k$-cycle to negative $k$-cycle in a $k$-circle-layered graph. The latter problem is defined as follows: given a directed weighted graph, find a $k$-cycle of negative total edge weight. 

\begin{lemma}
	If the negative directed $k$-cycle problem on a $k$-circle-layered graph, $G$
	with weights in the range $[-R,R]$ can be solved in time $T(n,m,R)$ time then the directed minimum k-cycle problem with weights in the range $[-R,R]$ can be solved in  $\tO(\lg(R)T(n,m,R))$ time.
	\label{lem:negativeDirCycle}
\end{lemma}
\begin{proof}
We use the color coding lemma (Lemma \ref{lem:colorcode}) to reduce the directed minimum $k$-cycle problem to directed minimum $k$-cycle in a $k$-circle-layered graph. Let the partitions of this graph $G$ be $U_1$, $U_2$, \ldots , $U_k$.

Now, given any weight $T\in [-Rk,Rk]$, we can add $-T$ to all edges between $U_1$ and $U_2$, creating a new graph $G(T)$ , and ask for a negative $k$-cycle in $G(T)$.
If such a cycle $C$ exists, then since $C$ must go around the circle of the $k$-circle-layered graph, then the weight of $C$ in $G(T)$ is its weight $w(C)$ in $G$, minus $T$. Notice that $w(C)-T<0$ if and only if $w(C)<T$. Hence, for any $T$, using one negative $k$ cycle query we can check whether the Min Weight $k$ Cycle in $G$ has weight $<T$ or not.

We can now binary search over the possible cycle lengths $O(\lg(R))$ times to find the minimum cycle length in $G$.
\end{proof}

\begin{theorem}
	If Radius in an undirected weighted $N$ node $M$ edge graph can be computed in $f(N,M)$ time 
	then the minimum weight directed $k$-cycle problem in $n$ node, $m$ edge graphs with weights in $[-R,R]$ can be solved in $\tO(f(O(n),O(m))\log R+m)$ time.
	\label{thm:undRadiusisMN}
\end{theorem}
\begin{proof}
		We take the minimum weight $k$-cycle problem with edge weights in the range $[-R,R]$ and use our previous reduction (from Lemma \ref{lem:negativeDirCycle}) to the negative weight $k$-cycle problem in a $k$-circle-layered graph, $G$, with partitions $U_1$, $U_2$, \ldots , $U_k$. We will refer to the $j$th node in $U_i$ by $u_i^j$.
		We will create a new graph $G'$ on partitions $V_1,\ldots,V_k$ where $V_i$ corresponds to the nodes in $U_i$ of $G$. In particular, $v_i^j$ corresponds to node $u_i^j$ of $G$.
		See Figure \ref{fig:radius} for an illustration of the construction.
		\begin{figure*}[h]
			\centering
			\includegraphics[scale=1]{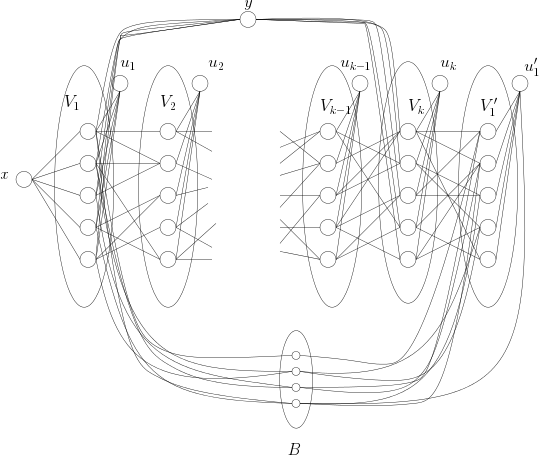}
			\caption{The Radius gadget.}
			\label{fig:radius}
		\end{figure*}
		
		Let $F = 20kR$.
		Let $V_1'$ be a copy of the nodes in $V_1$. 
		
		Let the edges from $v_i^{(j)} \in V_i$ to $v_{i+1}^{(p)} \in V_{i+1}$ for $i \in [1,k-1]$ exist if there is an edge between $u_i^{(j)} \in U_i$ and $u_{i+1}^{(p)} \in U_{i+1}$. The weight of the edge is given by  $w(v_i^{(j)},v_{i+1}^{(p)}) = w(u_i^{(j)},u_{i+1}^{(p)})+F$.
		
		We will introduce edges between $v_k^{(j)}\in V_{k}$ and ${v'}_{1}^{(p)} \in V'_1$ if there is an edge between $u_k^{(j)} \in U_k$ and $u_{1}^{(p)} \in U_{1}$. The weight of the edge is given by  $w(v_k^{(j)},{v'}_{1}^{(p)}) = w(u_k^{(j)},u_{1}^{(p)})+F$.
		
		Let $\hat{V_i}$ be $V_i \cup \{u_i\}$ where $u_i$ is a new node connected to all nodes in $V_i$ with edges of weight $3F/4= 15kR$.
		
		Let $B$ be a set of $\lg(n)+1$ nodes $b_0,b_1,\ldots,b_{\lg(n)}$ where each $b_i$ is connected to every node $v_1^{(j)}$ where the $i^{th}$ bit of $j$ is a $1$ and connect $b_i$ to every node ${v'}_1^{(j)}$ where  the $i^{th}$ bit of $j$ is a $0$. Here the weights of these edges are $kF/2$ between $V_1$ and $B$ and the weights are $kF/2-kR$ between $V'_1$ and $B$. Connect the node $u_1'$ node to all nodes in $B$. 
		
		We add a node $x$ which has edges of weight $kF-1$ to all nodes in $V_1$, so only nodes in $V_1$ or $x$ could possibly have radius $<kF$. Note that if $k>1$ then the shortest path from $x$ to nodes in $V_2$ is at least $kF-1+20kR-R$, making it impossible for it to be the center if the radius is less than $kF$.

		We add a node $y$ which has edges of weight $kF/2$ to all nodes in $V_1$ and edges of weight $kF/2-F/2$ to every node in $V_k$.

		So the shortest path between $v_1^{(i)}$ and ${v'}_1^{(j)}$ when $i \ne j$ is $kF-kR$.
		
		The shortest path between $v_1^{(i)}$ and ${v'}_1^{(i)}$ has a length of the shortest $k$-cycle, if one exists, through node $u_1^{(i)}$ plus $kF$. Because, all other choices of paths use more than $k$ edges (thus giving a weight of at least $(k+1/2)F$) or go though $y$ or $B$ which will result in a weight of $(k+1/2)F$ or more. So, the shortest path will be the shortest k-cycle through node $u_1^{(i)}$ plus $kF$.

		Any node $v_1^t$ in $V_1$ whose copy is part of any $k$-cycle in $G$ will have a path of length $<kF$  to all nodes in $V_i$ where $i<k$. To see this, first note that $v_1^t$ can reach some node in $V_i$ with a path of length at most $(i-1)F+(i-1)R=F(i-1)(1+1/(20k))$. Now, since every two node in $V_i$ are connected with a path of length $1.5F$ through node $u_i$, $v_1^t$ can reach any node in $V_i$ using a path of length at most 
		$(i-1+3/2+(i-1)/(20k))F$ which for $i<k$ is at max $(k-1/2+ (k-2)/(20k))F < (k-0.48)F < kF$.
		%
		%
		
		We need to ensure that any node in $V_1$ involved in a cycle has a path of length $<Fk$ to all nodes in $V_k$. The node $y$ solves this problem by providing a path of length $(k-0.5)F$ between all nodes in $V_1$ and $V_k$.

		So every node $v_1^{(i)}$ in $V_1$ has shortest paths of length $<kF$ to every node in the graph except, possibly, to the node ${v'}_1^{(i)}$. If $v_1^{(i)}$ is involved in a negative $k$-cycle then there is a shortest path between ${v'}_1^{(i)}$ and $v_1^{(i)}$ of length $<Fk$. Thus, if there is a negative $k$-cycle then the radius of the graph will be $<Fk$.
\end{proof}

We would like to show that radius in unweighted graphs is also hard from (unweighted) $k$-cycle. 

\begin{lemma}
If Radius in an undirected unweighted $N$ node $M$ edge graph can be computed in $f(N,M)$ time 
	then a directed $k$-cycle problem in $n$ node, $m$ edge graphs can be found in $\tO(f(O(n),O(m))+m)$ time.
%
	\label{lem:unweightedundRadiusisMN}
\end{lemma}
\begin{proof}
	If there is a $k$-cycle then the radius will be $k$ and otherwise the radius will be larger.
	The reduction is similar to before: we start with a $k$-circle-layered graph $G$ with partitions $U_1,\ldots,U_k$ and reduce it to a new graph $G'$ with roughly the same number of nodes. See Figure~\ref{fig:radiusunweight} for a depiction of this graph. In $G'$, we have partitions $V_1,\ldots,V_k$ where $V_i$ corresponds to $U_i$ in $G$.
	
	As before, let $V_1'$ be a copy of the nodes in $V_1$.

	Let the edges from $v_i^{(j)} \in V_i$ to $v_{i+1}^{(p)} \in V_{i+1}$ for $i \in [1,k-1]$ exist if there is an edge between $u_i^{(j)} \in U_i$ and $u_{i+1}^{(p)} \in U_{i+1}$. 
	
	We will introduce edges between $v_k^{(j)}\in V_{k}$ and ${v'}_{1}^{(p)} \in V'_1$ if there is an edge between $u_k^{(j)} \in U_k$ and $u_{1}^{(p)} \in U_{1}$.

	Let $\hat{V_i}$ be $V_i \cup u_i$ where $u_i$ is a node connected to all nodes in $V_i$. We connect $u_i$ to $u_{i+1}$ for all $i \ in [1,k-1]$ and connect $u_k$ and $u_1'$. We additionally add edges between $u_2$ and every node in $V_1$.

	Let $B$ be a set of $\lg(n)+1$ nodes $b_0,b_1,\ldots,b_{\lg(n)}$ where each $b_i$ is connected to every node $v_1^{(j)}$ where the $i^{th}$ bit of $j$ is a $1$.
	 Let $B'$ be a set of $\lg(n)+1$ nodes $b_0',b_1',\ldots,b_{\lg(n)}'$ where each $b_i'$ is connected to every node ${v'}_1^{(j)}$ where  the $i^{th}$ bit of $j$ is a $0$. We connect $b_i$ and $b'_i$ with a path of $k-2$ nodes.

	We add a node $x$ which to all nodes in $V_1$. We connect to $x$ a path $P$ containing nodes $p_1,\ldots,p_{k-1}$, so that nodes in $V_1$  are distance $k$ from $p_{k-1}$. So only nodes in $V_1$ or $x$ could possibly have radius $<k$.

			\begin{figure*}[h]
				\centering
				\includegraphics[scale=0.6]{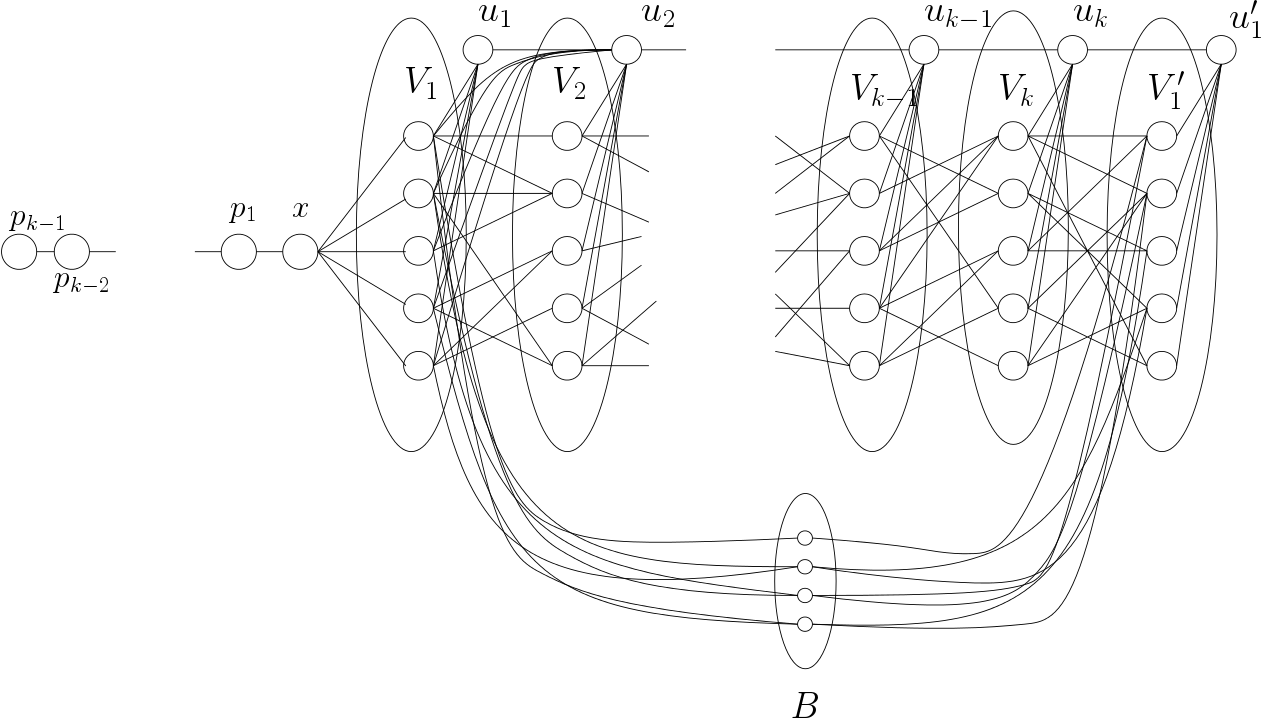}
				\caption{The unweighted radius gadget.}
				\label{fig:radiusunweight}
			\end{figure*}

	The shortest path between $v_1^{(i)}$ and ${v'}_1^{(i)}$ has length $k$ through $B$ and $B'$. There is a path to $u_i$ from any node in $V_1$ is $\leq k$. The shortest path from nodes in $V_1$ to nodes in $x \cup P$ is $\leq k$.

	The shortest path between a node in  $V_1$ and a node in $V_i$ has length at most  $i$ for $i>1$. Thus, every node in $v_1^{(i)} \in V_1$ is distance $\leq k$ from every node except, possibly ${v_1^{(i)}}'$.

	If a $k$-cycle exists then that $k$-cycle corresponds to a node such that ${v_1^{(i)}}'$ and ${v_1^{(i)}}$ are at distance $\leq k$. If a $k$-cycle does not exist then there is no path from $V_1$ to $V_2$ to $\ldots$ to $V_k$ and back to the corresponding node in $V_1'$, thus the shortest path between the two nodes has length $>k$. Paths through the nodes $u$ require extra edges and thus require at least $k+1$ edges to get to ${v_1^{(i)}}'$ from ${v_1^{(i)}}$.
	
	We can detect a $k$-cycle by running radius and returning true if the radius is $k$ or less and false otherwise. 
	
\end{proof}

%

%% file: wienerIndex.tex
\subsection{Undirected k-cycle reduces to Wiener index}
\label{subsec:wiener}
Informally, The Wiener index in a graph is the sum of all pairwise distances in the graph. Formally:

\begin{definition}
	Let $d(v,u)$ be the (shortest path) distance between $v$ and $u$ in $G$. Then the Wiener index is
	$$\mathbb{W}(G) =  \sum_{v \in V}\sum_{u\in V} d(v,u).$$
\end{definition}

\begin{theorem}
If the Wiener index in an undirected weighted graph on $N$ nodes and $M$ edges can be computed in $f(N,M)$ time  then the minimum weight $k$-cycle problem in $N$ node, $M$-edge directed graphs can be solved in $\tO(f(N,M)+M)$ time.
\label{thm:wienerIndex}
\end{theorem}
\begin{proof}
	We take the minimum weight $k$-cycle problem with edge weights in the range $[-R,R]$ and use our previous reduction (from Lemma \ref{lem:negativeDirCycle}) to the negative weight $k$-cycle problem on a $k$-partite graph, $G$, with partitions $U_1$, $U_2$, \ldots , $U_k$ such that edges only exist between partitions $U_i$ and $U_{(i+1 \mod{k})}$. We will describe the reduction graph $G'$, see Figure \ref{fig:wiener} for a diagram of the gadget. $G'$ contains a part $V_i$ for each part $U_i$ of $G$, for each $i\in\{1,\ldots,k\}$. In addition, there is a part $V_1'$ which contains copies of the nodes in $V_1$. For each $i$, call the $j$th node in $V_i$, $v_i^j$.
		
\begin{figure}[h]
			\centering
			\includegraphics[scale=0.5]{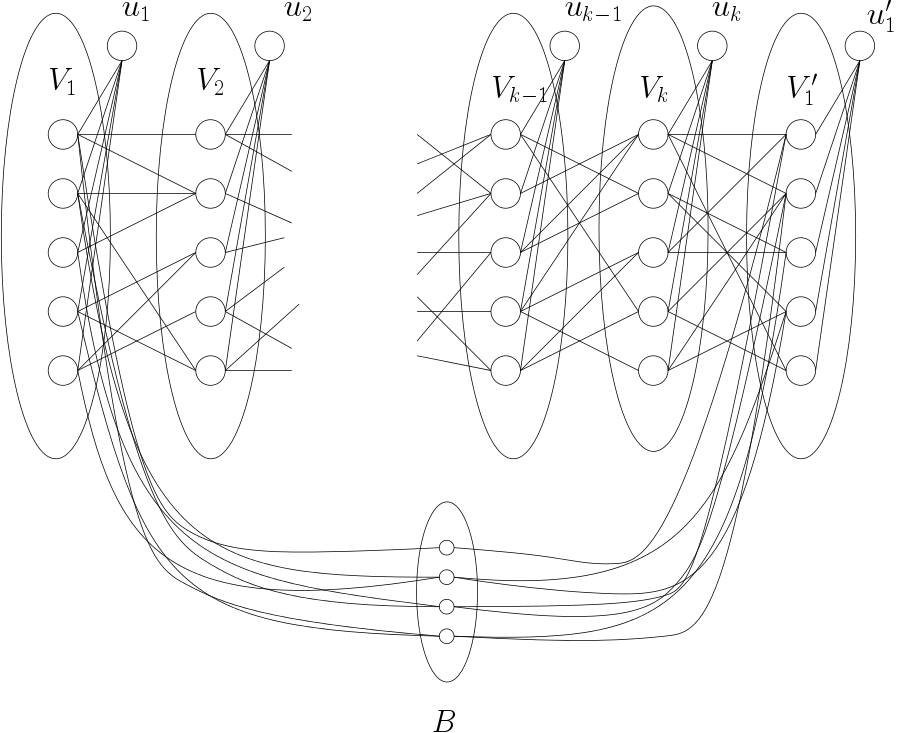}
			\caption{The Wiener Index gadget.}
			\label{fig:wiener}
\end{figure}
		
	Let $F = 20kR$.
	Let the edges from $v_i^{(j)} \in V_i$ to $v_{i+1}^{(p)} \in V_{i+1}$ for $i \in [1,k-1]$ exist if there is an edge between $u_i^{(j)} \in U_i$ and $u_{i+1}^{(p)} \in U_{i+1}$. The weight of the edge is given by  $w(v_i^{(j)},v_{i+1}^{(p)}) = w(u_i^{(j)},u_{i+1}^{(p)})+F$.

	We will introduce edges between $v_k^{(j)}\in V_{k}$ and ${v'}_{1}^{(p)} \in V'_1$ if there is an edge between $u_k^{(j)} \in U_k$ and $u_{1}^{(p)} \in U_{1}$. The weight of the edge is given by  $w(v_k^{(j)},{v'}_{1}^{(p)}) = w(u_k^{(j)},u_{1}^{(p)})+F$.

	Let $\hat{V_i}$ be $V_i \cup \{u_i\}$ where $u_i$ is a node connected to all nodes in $V_i$ with edges of weight $3F/4= 15kR$.	Let $B$ be a set of $\lg(n)+1$ nodes $b_0,b_1,\ldots,b_{\lg(n)}$ where each $b_i$ is connected to every node $v_1^{(j)}$ where the $i^{th}$ bit of $j$ is a $1$ and connect $b_i$ to every node ${v'}_1^{(j)}$ where  the $i^{th}$ bit of $j$ is a $0$. Where the weights of these edges are $kF/2$ between $V_1$ and $B$ and the weights are $kF/2-kR$ between $V'_1$ and $B$.
	By construction, the shortest path between $v_1^{(i)}$ and ${v'}_1^{(j)}$ when $i \ne j$ is $kF-kR$.
	
	Let $H$ be the graph composed of $\hat{V_2} \cup \ldots \cup \hat{V_k}$.\\
	Let $H'$ be the graph composed of $\hat{V_1} \cup H$.
	Let $H''$ be the graph composed of $H \cup \hat{V'_1}$.	
	Let $G'$ be the graph composed of $H \cup \hat{V'_1} \cup B$ where we add edges $\forall i \in [1,|V_1|]$ with weight $w(v_1^{(i)},{v'}^{(i)}_1)= kF$.

	Let the minimum $k$-cycle length in $G$ going through $u_1^{(i)}$ be of length $c_i$. Note that $d_{G'}(v_1^{(i)},{v'}^{(i)}_1) = kF+\min(0,c_i)$ by construction. Note that $d_{G'}(v_1^{(i)},{v'}^{(j)}_1) = kF-kR$ for $i \ne j$. If we could access the sum of all of the $|V_1|^2$ pairwise distances between the nodes in $V_1$ and $V'_1$ and see if those distances are, in sum, $(kF-kR)|V_1|^2+|V_1|kR$ then we would know if a negative $k$-cycle exists.
	
	We need to remove all the other distances first.
 Most easily we can calculate the distances to and from nodes in $B$. 
	Let $w_B = \sum_{b\in B}\sum_{v\in G'} \delta_{G'}(b,v)$, we compute this in time $O(|B||G'|)=O(\lg(N) |G'|)$.

	Next we can get the sum of the distances between any nodes in $H$ by computing $\mathbb{W}(H)$. The graph has $M+N = O(M)$ edges in total.
After this, we want to find the distances to and from the nodes in $\hat{V_1}$ to $G'/(\hat{V'_1}\cup B)$. We can get this from $\mathbb{W}(H')-\mathbb{W}(H)$.
Next we want to find the distances too and from the nodes in $\hat{V'_1}$ to $G'/(\hat{V_1}\cup B)$, we can get this from $\mathbb{W}(H'')-\mathbb{W}(H)$.
Finally we need to account for the distance between $u_1$ and $u'_1$, call this distance $w_u$ we can run Dijstra's algorithm and find this distance in $\tilde{O}(M+N) = \tilde{O}(M)$ time. 
	
	Therefore the sum of all the pairwise distances between nodes in $V_1$ and $V'_1$ is equal to 
	$\mathbb{W}(G')-(\mathbb{W}(H')-\mathbb{W}(H))-(\mathbb{W}(H'')-\mathbb{W}(H))-\mathbb{W}(H)-2w_B-2w_u.$
	
	Thus, if $\mathbb{W}(G')-\mathbb{W}(H')-\mathbb{W}(H'')+\mathbb{W}(H)-2w_B-2w_u<(kF-kR)|V_1|^2+|V_1|kR$ then there is a negative $k$-cycle. If $\mathbb{W}(G')-\mathbb{W}(H')-\mathbb{W}(H'')+\mathbb{W}(H)-2w_B-2w_u \geq (kF-kR)|V_1|^2+|V_1|kR$ then there is no negative $k$-cycle.
	
	The total time to get all the sums we need is $O(f(N,M+N)+M\lg(N)) = \tO(f(N,M)+M\lg(N))$.
	
\end{proof}

We would like to show that {\em unweighted} Wiener Index is also hard from $k$-cycle. 

\begin{lemma}
	If the Wiener Index can be computed in $f(N,M)$ time in an  undirected unweighted  graph then the directed $k$-cycle problem can be solved in $\tO(f(N,M)+M)$ time.
	\label{lem:unweightedundWeinerMN}
\end{lemma}
\begin{proof}
	If there is a $k$-cycle then the Wiener Index will be lower and otherwise the Wiener Index will be larger. We will form the graph $G'$ similar to before. See Figure~\ref{fig:wienerunweight} for a depiction of this graph. Let the $U_i$ and $V_i$ and $V_1'$ be as in the previous proof. 
	
	Let the edges from $v_i^{(j)} \in V_i$ to $v_{i+1}^{(p)} \in V_{i+1}$ for $i \in [1,k-1]$ exist if there is an edge between $u_i^{(j)} \in U_i$ and $u_{i+1}^{(p)} \in U_{i+1}$. \\
	We will introduce edges between $v_k^{(j)}\in V_{k}$ and ${v'}_{1}^{(p)} \in V'_1$ if there is an edge between $u_k^{(j)} \in U_k$ and $u_{1}^{(p)} \in U_{1}$. \\
	Let $\hat{V_i}$ be $V_i \cup u_i$ where $u_i$ is a node connected to all nodes in $V_i$. We connect $u_i$ to $u_{i+1}$ for all $i \ in [1,k-1]$ and connect $u_k$ and $u_1'$. We additionally add edges between $u_2$ and every node in $V_1$.\\
	Let $B$ be a set of $\lg(n)+1$ nodes $b_0,b_1,\ldots,b_{\lg(n)}$ where each $b_i$ is connected to every node $v_1^{(j)}$ where the $i^{th}$ bit of $j$ is a $1$.
	Let $B'$ be a set of $\lg(n)+1$ nodes $b_0',b_1',\ldots,b_{\lg(n)}'$ where each $b_i'$ is connected to every node ${v'}_1^{(j)}$ where  the $i^{th}$ bit of $j$ is a $0$. We connect $b_i$ and $b'_i$ with a path of $k-2$ nodes. \\
	
	\begin{figure}[h]
		\centering
		\includegraphics[scale=1]{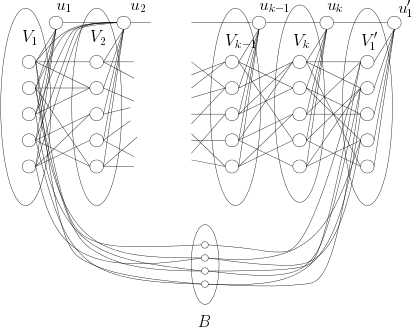}
		\caption{The unweighted Wiener Index gadget.}
		\label{fig:wienerunweight}
	\end{figure}

	The shortest path between $v_1^{(i)}$ and ${v'}_1^{(i)}$ has length $k$ through $B$ and $B'$. There is a path to $u_i$ from any node in $V_1$ is $\leq k$. The shortest path from nodes in $V_1$ to nodes in $x \cup P$ is $\leq k$. \\
	A the shortest path between a node in  $V_1$ and a node in $V_i$ has length at most  $i$ for $i>1$. Thus, every node in $v_1^{(i)} \in V_1$ is distance $\leq k$ from every node except, possibly ${v_1^{(i)}}'$.\\
	
	If a $k$-cycle exists then that $k$-cycle corresponds to a node such that ${v_1^{(i)}}'$ and ${v_1^{(i)}}$ are at distance $\leq k$. If a $k$-cycle does not exist then there is no path from $V_1$ to $V_2$ to $\ldots$ to $V_k$ and back to the corresponding node in $V_1'$, thus the shortest path between the two nodes has length $>k$. Paths through the nodes $u$ require extra edges and thus require at least $k+1$ edges to get to ${v_1^{(i)}}'$ from ${v_1^{(i)}}$.
	
	
	Let $H' = G'/(\hat{V_1'} \cup \hat{V_1} \cup B \cup B')$, $H'' = G'/(\hat{V_1'} \cup B \cup B')$ and $H''' =  G'/(\hat{V_1} \cup B \cup B')$.\\
	
	There are  $O(k\lg(n))$ nodes in $B$, $B'$ and the path between them. We can calculate these distance in $O(\lg(n)m)$ time. Call these distances $d_B$.
	
	Similarly calculate the distance from nodes $u_1$ to $\hat{V_1'}$ and $u_1'$ to $\hat{V_1}$ graph and call the sum of these two values $d_u$.
	
	Note that $\mathbb{W}(H')$ captures the distances between nodes in $H'$, the shortest paths between them never use nodes outside of $H'$. Similarly $\mathbb{W}(H'')$ and $\mathbb{W}(H''')$ each capture the distances between the nodes contained in them . 
	
	 $\mathbb{W}(G')-\mathbb{W}(H'')-\mathbb{W}(H''')+\mathbb{W}(H')-2d_B-2d_u$ is equal to the sum of $\sum_{u\in V_1} \sum_{v \in V_1'} d(u,v)$. 
	
	Now let 
	$d_x = \sum_{v_1^{(i)}\in V_1} \sum_{{v_1^{(j)}}' \in V_1' \text{ and } i\ne j} d(v_1^{(i)},{v_1^{(j)}}') $. And note that $d_x = |V_1|(|V_1|-1) k$.
	
	Next note that  $d(v_1^{(i)},{v_1^{(i)}}')=k+1$ if there is no $k$-cycle through  $v_1^{(i)}$. 
	
	So if there is no $k$-cycle then $\sum_{u\in V_1} \sum_{v \in V_1'} d(u,v) = |V_1|^2k+|V_1|$ and if there is a $k$-cycle the sum is less. 
	
	There is a k-cycle if $\mathbb{W}(G')-\mathbb{W}(H'')-\mathbb{W}(H''')+\mathbb{W}(H')-2d_B-2d_u < |V_1|^2k+|V_1|$ and there is no k-cycle if $\mathbb{W}(G')-\mathbb{W}(H'')-\mathbb{W}(H''')+\mathbb{W}(H')-2d_B-2d_u= |V_1|^2k+|V_1|$.
\end{proof}

\begin{corollary}
	If undirected (unweighted/weighted) Wiener Index can be computed in $f(N,M)$ time on a graph of density $M = \tTheta(N^{1+1/L})$ then the minimum weight (unweighted/weighted) k-cycle problem can be solved in $\tO(f(N,M)+M)$.
	\label{cor:weinerDirUndir}
\end{corollary}

%% file: APSP.tex
\subsection{Undirected k-cycle reduces to APSP}
A folklore reduction reduces Shortest Cycle in directed graphs to APSP in directed graphs as follows. Let $G$ be the graph in which we want to find the Shortest Cycle. Compute APSP in $G$, and then for every edge add $d(u,v)+w(v,u)$ and take the minimum - this is the weight of the Shortest Cycle in $G$. Reducing directed Shortest Cycle to APSP in undirected graphs seems more problematic, as noted by Agarwal and Ramachandran~\cite{agarwal2016fine}. In our paper, however, we were able to reduce Min Weight $k$-Cycle in a directed graph to Radius in undirected graphs. Radius in an undirected graph of course can easily be reduced to APSP: compute APSP and then set the radius to $\min_u\max_v d(u,v)$. Because of this, we immediately get the same lower bounds for APSP as for Radius and Min Weight $k$-Cycle. This reduction also works for unweighted graphs.


\begin{corollary}
If there is an $O(n^2+mn^{1-\eps})$ time algorithm for $\eps>0$ for APSP in $m$ edge $n$ node (directed/undirected, weighted/unweighted) graphs, then there is also an $O(n^2+mn^{1-\eps})$ time algorithm for Radius in $m$ edge $n$ node (directed/undirected, weighted/unweighted) graphs.
	\label{cor:APSPhard}
\end{corollary}

%% file: GeneralCSP.tex

\label{sec:GeneralCSP}

We will give the definitions of two versions of the degree-$k$-CSP problem. The maximization version and the exact weight version. Then we will show how maximum or exact weight $\ell$-hyperclique finding solves these problems using a generalization of the argument from Williams \cite{thesis}.  
Finanlly, we will give a reduction between the unweighted  $(\ell,k)$-hyperclique problem and the degree-$k$-CSP problems.


\begin{definition}
	The maximum degree-$k$-CSP problem has a formula with $n$ variables and $m$ clauses. Each clause is a formula on $n$ variables (though it may of course depend on fewer). For each clause $c_i$ $i\in[1,m]$ there exists a $d$ degree polynomial $p_i: \{0,1\}^n \rightarrow \{0,1\}$ such that $c_i$ is satisfied on an assignment $\vec{a}$ when $p_i(\vec{a}) = 1$ and unsatisfied when $p_i(\vec{a}) = 0$.
	
	The solution to a maximum $k$-CSP problem is a an assignment of the $n$ variables, $\vec{a} = a_1,\ldots,a_n$ such that  $\sum_{i=1}^m p_i(\vec{a})$ is maximized.
\end{definition}

\begin{definition}
A weighted degree-$k$-CSP problem has a formula with  $n$ variables and $m$ clauses and two special constants $K_v,K_p \in [poly(n)]$. Each clause is a formula on $n$ variables (though it may of course depend on fewer). For each clause $c_i$ $i\in[1,m]$ there exists a $k$ degree polynomial $p_i: \{0,1\}^n \rightarrow \{0,1\}$ such that $c_i$ is satisfied on an assignment $\vec{a}$ when $p_i(\vec{a}) = 1$ and unsatisfied when $p_i(\vec{a}) = 0$.

The solution to a weighted $k$-CSP is a an assignment of the $n$ variables, $\vec{a} = a_1,\ldots,a_n$ such that
$\sum_{i=1}^n a_i = K_v$ and 
$\sum_{i=1}^m p_i(\vec{a}) = K_p$. \cite{thesis}
\end{definition}

We will prove that finding a fast algorithm which breaks the $(\ell,k)$-hyperclique hypothesis (for $\ell$ and $k$ constant) results in algorithms for the maximum degree-$k$-CSP problem and weighted degree-$k$-CSP problem that run in $2^{(1-\eps)n}$ for some $\eps>0$. 

\subsection{Degree-k-CSP to Weighted Hyperclique}
We will start by proving that the degree $g$ coefficients of the polynomials $p_i$ are bounded by $2^{g-1}$. We will generalize the argument of R. Williams \cite{thesis}.

\begin{lemma}
Let $p$ be a degree $k$ polynomial that computes a boolean function. That is $p:\{0,1\}^n \rightarrow \{0,1\}$. 
The degree $g$ coefficients of $p$ are bounded by $[-2^{g-1},2^{g-1}]$ if $g>0$ and the degree zero coefficient must be in $\{0,1\}$.
\end{lemma}
\begin{proof}
R. Williams shows that this is true for $g=\{0,1,2\}$ for all $n$ \cite{thesis}. We will prove this by induction. 

If we have proven the statement for all coefficients of degree $g-1$  or less for all polynomials with $n<h+1$ and for coefficients of degree $g$ for all polynomials with $n<h$ then we will show that the statement is true for coefficients of degree $g$ and $n=h$. 
Let $q= p(x_1,\ldots,x_{h-1},1)$ and $r= p(x_1,\ldots,x_{h-1},0)$. Note $x_h \in \{0,1\}$ so $p= q$ when $x_h=1$ and $p = r$ when $1-x_h = 1$. 
 Thus, the polynomial can be split into  $p(x_1,\ldots,x_h) = x_hq(x_1,\ldots,x_{h-1}) + (1-x_h)r(x_1,\ldots,x_{h-1})$. 

Let $q_i = $ the sum of all the $i^{th}$ degree terms of $q$. Let $r_i = $ the sum of all the $i^{th}$ degree terms of $r$. Let $p_i = $ the sum of all the $i^{th}$ degree terms of $p$.  Note that both $r$ and $q$ have only $h-1$ variables, and thus the condition has been proven for them by assumption. 

We have  that $p_i = x_h(q_{i-1}-r_{i-1})+r_{d}$.  Note that if a $g$ degree term in $p$ does not include variable $x_h$ then the term comes from  $r_d$, thus, the condition holds because $r$ is degree $g$ and has fewer then $h$ variables. If a $g$ degree term includes $x_h$ then it comes from $x_h(q_{g-1}-r_{g-1})$, thus its coefficient is formed by the subtraction of two coefficients bounded by $[-2^{g-2},2^{g-2}]$ because $q_{g-1}$ and $r_{g-1}$ have degree less than $g$ and less than $h$ variables. The value of this coefficient is thus bounded by $[-2^{g-1},2^{g-1}]$. This covers all the $g$ degree terms in $p$. Thus, the condition holds for $g$ degree terms in polynomials with $h$ variables. 

Our base case is built from $g=\{0,1,2\}$ where this has been proven for all $n$ and the fact that for all $g > 2$ and $h<3$ the result is true because with at most $2$ variables the maximum degree is $2$ when $p$ is multi-linear. Note that $p$ must be multi-linear because $x_i \in \{0,1\}$ and thus $x_i^2 =x_i$.

By induction over $g$ and $h$ we prove the statement.
\end{proof}

\begin{lemma}
	A degree-$k$-CSP formula can be transformed into a $k$-uniform hypergraph, $H$, on $\ell2^{n/\ell}$ nodes ($\ell>k$) in $O(mn^k2^{kn/\ell})$ time with two weights per edge $W_1(e)$ and $W_2(e)$. Such that for every $\ell$-hyperclique, $h$, the $W_1$ weight (sum over all edges of  $W_1(e)$) is $\sum_{i=1}^m p_i(\vec{a})$ and the $W_2$ weight (sum over all edges of  $W_2(e)$) is $\sum_{i=1}^n a_i$   for some assignment $\vec{a}$. Furthermore, for every assignment  $\vec{a}$ there exists some $\ell$-hyperclique whose weight corresponds to it. The weights $W_1(e)$ are bounded by $[-mn^k2^k, mn^k2^k]$. The weights $W_2(e)$ are bounded by $n$. When $\ell$ and $k$ are constants.
	\label{lem:cspHypercliqueBase}
\end{lemma}
\begin{proof}
	We will split the variables $x_1, \ldots , x_n$ into $\ell$ groups $X_i = [x_{1+in/\ell}, \ldots, x_{(i+1)n/\ell}] $. We create a node for all $2^{n/\ell}$ assignments for each $X_i$ call this set $V_i$. This gives a total of $\ell2^{n/\ell}$ nodes. 
	
	Let $S[i_1,\ldots,i_k]$ be the set of monomials from $p_1,\ldots,p_m$ that sets $X_{i_1},\ldots,X_{i_k}$ are responsible for. We will consider set $S[i_1,\ldots,i_k]$ to be responsible for a monomial $c$ if $S[i_1,\ldots,i_k]$ is the first set in lexicographical order that contains all the variables from $c$. Note that $S[i_1,\ldots,i_k]$ may contain variables from all $m$ $p_i$s. Every monomial will be in some $S[i_1,\ldots,i_k]$ because every monomial involves at most $k$ variables.
	
	Let $T[i_1,\ldots,i_k]$ be the set of indices of the sets of variables that an edge covering  $X_{i_1},\ldots,X_{i_k}$ are responsible for. We will have $j \in T[i_1,\ldots,i_k]$ if $j \in \{i_1,\ldots,i_k\}$ and $T[i_1,\ldots,i_k]$ is the first $T$ in lexicographic order where $j \in \{i_1,\ldots,i_k\}$. 
	
	We create a hyper-edge for each choice of $k$ nodes with at most 1 node from each $V_i$ (we can not give two assignments to the same variable). We give an edge $e = (v_{i_1},\ldots, v_{i_k})$ weight 
	$$W_1(e) = \sum_{c \in S[i_1,\ldots,i_k]} c(v_{i_1},\ldots, v_{i_k})$$
	and 
	$$W_2(e)= \sum_{j \in T[i_1,\ldots,i_k]} \sum_{t =1}^{n/\ell} x_{t + jn/\ell}.$$
	
	We ideally want our hyper clique to be formed by all $\binom{\ell}{k}$ choices of $k$ elements from $\{v_1,v_2,\ldots,v_\ell\}$ where $v_i\in V_i$ is an assignment of the variables in $X_i$. If we choose two $v_i,v_i' \in V_i$ to be part of our clique there will be no edge that contains both by construction, thus there will be no hyperclique. 
	
	Note that this means that every $k$-uniform $\ell$-hyper-clique effectively gives an assignment to every variable. The clique has the weight of every monomial in $P= \sum_{i=1}^m p_i$ represented in $W_1$ exactly once. And, the clique has the weight of each variable represented in $W_2$ exactly once. 
	
	An edge could be responsible for at most $mn^k$ possible monomials each of which have a coefficient at most $2^k$. So the range of possible $W_1$ weights is  $[-mn^k2^k,mn^k2^k]$.
	
	Each edge can be responsible for at most $kn/\ell$ variables for a range of possible $W_2$ weights of $[0,kn/\ell]$. Both $k$ and $\ell$ are constants and $k<\ell$ so for simplicity we will simply bound $W_2$ by $n$.

	It takes $\ell2^{n/\ell}$ time to produce all the nodes. Every edge takes $m\binom{n}{k}$ time to produce the weights for at most. There are $O(\ell\binom{\ell}{k}2^{kn/\ell})$ edges. This takes a total of $O(\ell m\binom{n}{k}\binom{\ell}{k}2^{kn/\ell})$ time. Removing the constant factors we get $O(mn^k2^{kn/\ell})$ time.
\end{proof}

We will now use this lemma to show that maximum $\ell$-hyperclique  finding and  exact weight  $\ell$-hyperclique finding solve maximum degree-$k$-CSP and weighted degree-$k$-CSP respectively.  

\begin{lemma}
	If the maximum  $\ell$-hyperclique problem on a $k$-uniform hypergraph with weights in the range $[-mn^{k}2^k, mn^{k}2^k]$ can be solved in time $T(n)$ then the maximum degree-$k$-CSP problem can be solved in $O(T(\ell2^{n/\ell})+mn^k\ell2^{kn/\ell})$ time. When $\ell$ and $k$ are constants.
	\label{lem:maxHypercliqueMaxCSP}
\end{lemma}
\begin{proof}
	We use Lemma \ref{lem:cspHypercliqueBase} to transform the  maximum degree-$k$-CSP problem into a $k$-uniform hypergraph, $H$, on $\ell 2^{kn/\ell}$ nodes in $O(mn^k\ell2^{kn/\ell})$ time. Next we create $H'$ by taking every edge, $e$, from $H$ and assigning it weight $W_1(e)$. Then we run maximum $k$-hyperclique on $H'$ the maximum hyperclique found will correspond to the assignment $\vec{a}$ that maximizes $\sum_{i=1}^m p_i(\vec{a})$.
	
\end{proof}

\begin{lemma}
If the exact weight  $\ell$-hyperclique problem on a $d$-uniform hypergraph with weights in the range $[-mn^{k+1}2^k\binom{\ell}{k}/\ell, mn^{k+1}2^k\binom{\ell}{k}/\ell]$ can be solved in time $T(n)$ then the weighted degree-$k$-CSP problem can be solved in $O(T(\ell2^{n/\ell})+mn^k\ell2^{kn/\ell})$ time. When $\ell$ and $k$ are constants.
\label{lem:exactWeightcliqueWeightCSP}
\end{lemma}
\begin{proof}
	
	We use Lemma \ref{lem:cspHypercliqueBase} to transform the  maximum degree-$k$-CSP problem into a $k$-uniform hypergraph, $H$, on $\ell2^{n/\ell}$ nodes in $O(mn^k\ell2^{kn/\ell})$ time. Next we create $H'$ by taking every edge, $e$, from $H$ and assigning it weight $W_1(e)2 \binom{\ell}{k}n/\ell + W_2(e)$ in $H'$. Note that  the weight of $W_2$ is upper-bounded by $n$ and there are $\binom{\ell}{k}$ edges in a hyper-clique so the low order bits correspond to $\sum_{i=1}^n a_i$  and the high order bits correspond to $\sum_{i=1}^m p_i(\vec{a})$.
	
	Then we run exact weight $\ell$-hyperclique on $H'$ looking for weight $K_g =  K_p2 \binom{\ell}{k}n/\ell + K_v$ if we find a solution it will correspond to the assignment $\vec{a}$ that achieves  $\sum_{i=1}^m p_i(\vec{a}) = K_p$ and $\sum_{i=1}^n a_i = K_v$.
\end{proof}

\subsection{Degree-k-CSP to Unweighted Hyperclique}
We will now extend the reduction to the unweighted version of $(\ell,k)$ hyperclique. This reduction will introduce a large polynomial overhead. However, for many choices of degree-$k$-CSPs the best known algorithms for both the maximization and weighted variants require exponential time. Notably, for max-$3$-SAT no $2^{(1-\eps)n}$ algorithm is known. 

\begin{theorem}
	If the unweighted $\ell$-hyperclique problem on a $k$-uniform hypergraph
	can be solved in time $T(n)$ then the maximum degree-$k$-CSP problem can be solved in $O((mn^{k})^{\binom{\ell}{k}}T(2^{n/\ell})+mn^k2^{kn/\ell})$ time. When $\ell$ and $k$ are constants.
	\label{thm:maxHypercliqueMaxCSP}
\end{theorem}
\begin{proof}
	
	We apply Lemma \ref{thm:maxHypercliqueMaxCSP} and generate a instance of the   $\ell$-hyperclique problem on a $k$-uniform hypergraph with weights in the range $[-mn^{k}2^k, mn^{k}2^k]$ in time $O(mn^k2^{kn/\ell})$ time. 
	
	We want an unweighted version. We can do this by guessing the weights of the edges in our hyperclique and deleting edges that don't have these weights. In generating the hyper graph we will split the variables $x_1, \ldots , x_n$ into $k$ groups $X_i = [x_{1+in/\ell}, \ldots, x_{(i+1)n/\ell}] $. We create a node for all $2^{n/\ell}$ assignments for each $X_i$ call this set $V_i$. We can consider the edges $E[i_1,\ldots,i_k]$ which have exactly one vertex from each of $V_{i_1},\ldots,V_{i_k}$.
	
	We will make a guess about the weight of the $E[i_1,\ldots,i_k]$ edges, call the weight we guess $g[i_1,\ldots,i_k]$. There are total of $\binom{\ell}{k}$ edges in a hyperclique, we will make guesses for each edge. The guesses, $g$ of the weights can be limited to  the range  $[-mn^{k}2^k, mn^{k}2^k]$. This gives a total of $O( (m n^k)^{\binom{\ell}{k}})$ values we have to guess. Any hyperclique we find in a graph where we delete all $E[i_1,\ldots,i_k]$ edges have weight $g[i_1,\ldots,i_k]$ will have a total weight equal to the sum of all the guessed $g[i_1,\ldots,i_k]$. Formally, if $\binom{\{1,\ldots, \ell\}}{k}$ is the set of all $k$ element subsets of $\{1,\ldots, \ell\}$ then the total weight of any hyperclique is $\sum_{\{i_1,\ldots,i_k\} \in \binom{\{1,\ldots, \ell\}}{d}} g[i_1,\ldots,i_k]$.
	
	To solve the maximum $k$-degree CSP we exhaustively search through assignments of $g_{\{1,\ldots,k\}},\ldots,g_{\{i_0,\ldots,i_{k-1}\}},\ldots,$ $ g_{\{\ell-k,\ldots, \ell\}}$ in order from the maximum sum of guesses to the minimal sum. We delete all edges where the weight of the edge disagrees with the guess and then remove the weight from all the remaining edges. The first instance of $\ell$-hyperclique that finds a $\ell$-hyperclique has found the maximum weight $\ell$-hyperclique. 
	
	Thus, we can solve $(mn^k)^{\binom{\ell}{k}}$ instances of  $\ell$-hyperclique in a $k$ regular hypergraph with $N = \ell2^{n/\ell}$ nodes and $O(2^{kn/\ell})$ hyperedges to solve maximum degree-$k$-CSP. This takes time $\tO( (mn^k)^{\binom{\ell}{k}}T(2^{n/\ell})+  2^{nk/\ell})$.
\end{proof}

\begin{theorem}
	If the unweighted $\ell$-hyperclique problem on a $k$-uniform hypergraph 
	can be solved in time $T(n)$ then the weighted degree-$k$-CSP problem can be solved in $O((mn^{k+1})^{\binom{\ell}{k}}T(2^{n/\ell})+mn^k2^{kn/\ell})$ time. 
	\label{thm:exactWeightcliqueWeightCSP}
\end{theorem}
\begin{proof}
We apply Lemma \ref{thm:exactWeightcliqueWeightCSP} and generate a instance of the   $\ell$-hyperclique problem on a $k$-uniform hypergraph with weights in the range $[-mn^{k+1}2^k, mn^{k+1}2^k]$ in time $O(mn^k2^{kn/\ell})$ time and a target sum $K_g$. 
		
We want an unweighted version. We can do this by guessing the weights of the edges in our hyperclique and deleting edges that don't have these weights. In generating the hyper graph we will split the variables $x_1, \ldots , x_n$ into $k$ groups $X_i = [x_{1+in/\ell}, \ldots, x_{(i+1)n/\ell}] $. We create a node for all $2^{n/\ell}$ assignments for each $X_i$ call this set $V_i$. We can consider the edges $E[i_1,\ldots,i_k]$ which have exactly one vertex from each of $V_{i_1},\ldots,V_{i_k}$.
		
We will make a guess about the weight of the $E[i_1,\ldots,i_k]$ edges, call the weight we guess $g[i_1,\ldots,i_k]$. There are total of $\binom{\ell}{k}$ edges in a hyperclique, we will make guesses for each edge. The guesses, $g$ of the weights can be limited to  the range  $[-mn^{k}2^k, mn^{k}2^k]$. This gives a total of $O( (m n^k)^{\binom{\ell}{k}})$ values we have to guess. Any hyperclique we find in a graph where we delete all $E[i_1,\ldots,i_k]$ edges have weight $g[i_1,\ldots,i_k]$ will have a total weight equal to the sum of all the guessed $g[i_1,\ldots,i_k]$. Formally, if $\binom{\{1,\ldots, \ell\}}{k}$ is the set of all $k$ element subsets of $\{1,\ldots, \ell\}$ then the total weight of any hyperclique is $\sum_{\{i_1,\ldots,i_k\} \in \binom{\{1,\ldots, \ell\}}{k}} g[i_1,\ldots,i_k]$.
		
To solve the weighted $k$-degree CSP we exhaustively search through assignments of $g_{\{1,\ldots,k\}},\ldots,g_{\{i_0,\ldots,i_{k-1}\}},\ldots,$ $ g_{\{\ell-k,\ldots, \ell\}}$ where the sum of the guesses equals our target sum $K_g$. We delete all edges where the weight of the edge disagrees with the guess and then remove the weight from all the remaining edges. The first instance of $\ell$-hyperclique that finds a $\ell$-hyperclique has found a exact weight $\ell$-hyperclique. 
		
Thus, we can solve $(mn^{k+1})^{\binom{\ell}{k}}$ instances of  $\ell$-hyperclique in a $k$ regular hypergraph with $N = \ell2^{n/\ell}$ nodes and $O(2^{kn/\ell})$ hyperedges to weighted degree-$k$-CSP. This takes time $\tO( (mn^{k+1})^{\binom{\ell}{k}}T(2^{n/\ell})+  2^{nk/\ell})$.
\end{proof}

Note that if $\ell>>k$ then the time to solve the hyperclique problem overwhelms the time it takes to produce the graph. Further note that Max $\ell$-SAT is a special case of the maximum degree-$\ell$-CSP problem.

\begin{corollary}
An algorithm which breaks the $(\ell,k)$-hyperclique hypothesis runs in time $n^{\ell-\eps}$ for some $\eps>0$. 

When $\ell$ and $k$ constant
such an algorithm can be used to solve the maximum degree-$k$-CSP problem and weighted degree-$k$-CSP problems in time $2^{(1-\eps/\ell -o(1))n}$. 
\end{corollary}

%% file: nonIntDensity.tex
We get tight $mn$ hardness for graphs with densities of the form $m=n^{1+1/L}$ for integer $L\geq 1$. For graphs with densities between $m  = \ttildeo( n^{1+1/L})$ and $ m = \tilde{\omega}(1+1/(L+1))$ we get the lower bound of $n^{2+1/(L+1)}$. So our lower bound looks like a staircase (see Figure \ref{fig:lbimg2}). The best known upper bound for all of these problems is $mn$. Ideally there would be a tight lower bound between the points $m=n^{1+1/L}$ and $m=n^{1+1/(L+1)}$. This does not exist right now, but we can present an improved lower bound for some densities in between $m=n^{1+1/L}$ and $m=n^{1+1/(L+1)}$. As a result we can rule out the ``stair case'' lower bound, making improvements between these densities even less likely.

\begin{figure}[h]
		\centering
		\includegraphics[scale=0.4]{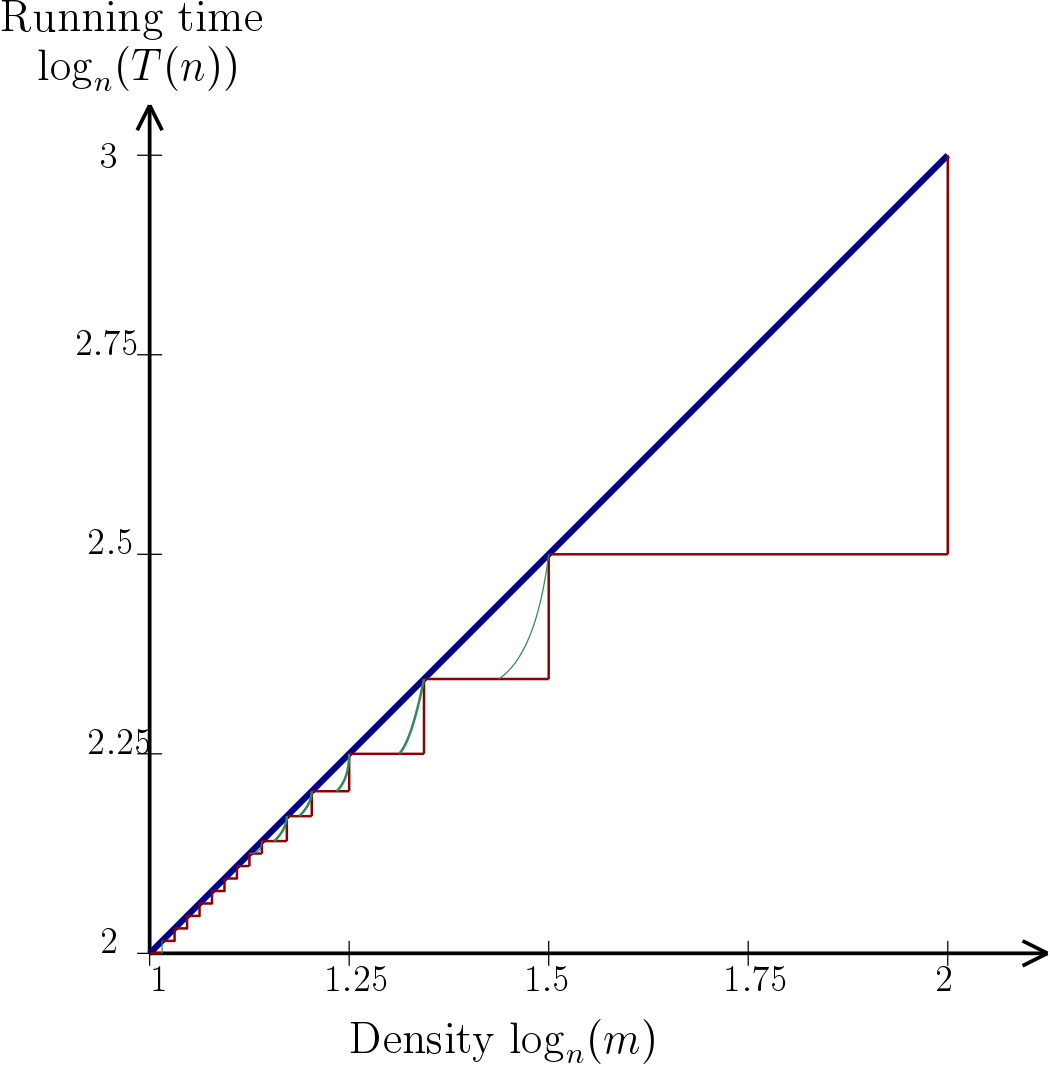}
		\caption{The heavy blue line represents the upper bound of $\tO(nm)$ and the jagged lighter red line represents the lower bound at all densities $m = \tTheta(n^{1+1/L})$. The curved green lines show the lower bounds from appendix \ref{sec:nonIntgerDensities}.}
		\label{fig:lbimg2}
\end{figure}

At a high level we will use random color coding on the vertices with $c$ colors. This will cause the node sets to be of size approximately $n/c$. The edge sets will effectively be colored by pairs of these colors (the color of their endpoints) for a total of $c^2$ colors. So the number of edges between nodes of two specifically chosen colors is in expectation $m/c^2$. If the graph is sparse to begin with then this coloring will cause the graph to become increasingly sparse. However, this reduction from denser to less dense graphs is inefficient; we need to consider $c^k$ possible choices of colors of vertices. For densities close to $m = n^{1+1/L}$ we will get a better upper bound than before. Sadly this approach does not work for $L=1$ because color coding a dense graph generates many dense graphs.

We will re-phrase and use Lemma 5 from Pagh and Silvestri's paper.
\begin{theorem}
We randomly color a graph's nodes with $c$ colors. Let the set of edges between color $c_i$ and $c_j$ be called $E_{i,j}$. Then 
$P( |E_{i,j}| > Ex/c^2) \leq 1/x^2.$
\cite{pagh2014input}
\label{thm:pagh2014input}
\end{theorem}

We apply this to get a self-reduction from $k$-cycle to $k$-cycle at a different density.

\begin{lemma}
	If we can solve the $k$-cycle problem in $T(n,m)$ time then for all $g \in [1,n]$ and all $c \in [1,n]$:
	$$T(n,m) = O(nm\lg(n)/g + nm\lg(n)/c + c^k T(n \lg(n)/c, m\sqrt{g}/ c^{1.5})).$$
	\label{lem:selfRedDens}
\end{lemma}
\begin{proof}
	First we take the graph $G$ and turn it into a $k$-circle-layered graph using Lemma \ref{lem:colorcode}. 
	Then we check for all high degree nodes with degree greater than $gm/n$ to find what the minimum $k$-cycle through them is. There are $n/g$ such nodes at most and we can use Dijstra to find the minimum $k$-cycle in time $nm/g$. Remove all the high degree nodes from the graph and call this new graph $G'$.
	
	We randomly color the nodes of the altered graph $G'$ where $|V'|=n$ and $|E'|=m$ with $c$ colors. Let $V_i$ be the set of vertices colored color $i$. Let $E_{i,j}$ be the number of edges between vertices colored $i$ and $j$. 
	
	There are $O(n\lg(n)/c)$ nodes in sets $V_i > n\lg(n)/c$ in expectation. In $O(nm\lg(n)/c)$ time we can remove all nodes in overloaded sets.  
	
	By Theorem \ref{thm:pagh2014input} we have that the number of sets with $E_{i,j}> m \sqrt{g}/c^{1.5}$ in expectation is $1/(cg)$ removing every node in a set attached to these edges takes at most $nm\lg(n)/g$ time. 
	
	We can find all $k$-cycles in the remaining graph by trying all $n^k$ choices of colorings of the $k$-cycle. Each subproblem has $kn\lg(n)/c$ nodes and $km\sqrt{g}/c^{1.5}$ edges. 

	The minimum k-cycle found in any of these steps can be returned. 
\end{proof}

Now we can state what our improved bound. It drops quickly away from the ideal $mn$ bound, but we beat the ``staircase'' density depicted in red in Figure \ref{fig:lbimg2}. This lower bound rules out the possibility that the true answer would be the staircase difficulty. 

\begin{lemma}
	Assuming the $k$-clique where $k>5$ and odd conjecture we have that for $c=n^\delta$ for all $\delta>0$ :
	$$T(n^{1+o(1)}/c, mn^{o(1)}/ c^{1.5})) = \Omega(mn^{1-o(1)}c^{-k}) .$$
	This bound beats the bound from $k+1$-clique when $\delta<4/((k-1)(k^2-k+4))$.
\end{lemma}
\begin{proof}
If the $k$-clique conjecture is true then there exists some $f(n)= n^{o(1)}$ such that $k$-clique takes $\omega(nm/f(n))$. Set the $g$ from Lemma \ref{lem:selfRedDens} equal to $f(n)\lg(n)$.

Then $k$-clique takes $\omega(nm\lg(n)/g + nm\lg(n)/c)$ because $c$ is a polynomially large factor. Then re-stating the formula from Lemma \ref{lem:selfRedDens} we get:
$$T(n^{1+o(1)}/c, mn^{o(1)}/ c^{1.5})) = \Omega(mn^{1-o(1)}c^{-k}) .$$

Let $N_c = n^{1+o(1)}/c$ and $M_c = mn^{o(1)}/ c^{1.5})$.
When $k$ odd we show $k$-cycle is $mn$ hard for density $m = n^{(k+1)/(k-1)} = n^{1-2/(k-1)}$ and $k+2$-cycle is $m'n$ hard for density  $m = n^{(k+3)/(k+1)}$.

Note that if $c = 1$ we get that $M = N^{(k+1)/(k-1)}$. Further note that if $k>5$ and $c=n^{8/((k-1)(k-3))}$ then $M = N^{(k+3)/(k+1) +o(1)}$. So with $c \in [1,n^{8/((k-1)(k-3))}]$ spans the distance between the two densities we care about for $k>5$. 
 
Next recall that the lower bound we have on graphs between density $m = n^{(k+1)/(k-1)}$ and $m = n^{(k+3)/(k+1)}$ is $\Omega(n^{2 + 2/(k+1)-o(1)})$. For what values of $c$ do we beat this lower bound? Well the question can be reformulated as: for what values of $c$ is $N^{2 + 2/(k+1)-o(1)} = o(mn^{1-o(1)}c^{-k})$. After some algebraic manipulation we find that $c = \tilde{o}\left(n^{4/((k-1)(k^2-k+4))}\right)$. So when $c = n^{\delta}$ and $k>5$ we have an improved lower bound when $\delta <  4/((k-1)(k^2-k+4))$. 

\end{proof}